\documentclass[aps,prl,reprint,superscriptaddress,longbibliography,amsmath,amssymb]{revtex4-2}

\usepackage{svg}
\usepackage{color}
\usepackage{graphicx}
\usepackage{times}
\usepackage{dcolumn}
\usepackage{bm}
\usepackage{physics}
\usepackage{hyperref}
\usepackage{url}
\usepackage{tikz-cd}
\usepackage{notes2bib}
\usepackage{multirow}
\usepackage{mathrsfs}
\usepackage{tabularx}
\usepackage{dsfont}
\usepackage{verbatim}
\usepackage{amsfonts}
\usepackage{amsthm}
\usepackage{mathtools}
\usepackage{qcircuit}
\usepackage[normalem]{ulem}
\usepackage{xcolor}
\usepackage{csquotes}
\usepackage[many]{tcolorbox}
\usepackage[clock]{ifsym}
\usepackage{fontawesome}
\usepackage{manfnt}
\usepackage{enumitem}
\usetikzlibrary{arrows.meta, calc}

\DeclarePairedDelimiterX{\barpair}[2]{(}{)}{%
  #1\;\delimsize\|\;#2%
}

\usepackage{stackengine,scalerel}

\theoremstyle{definition}
\newtheorem{theorem}{Theorem}
\newtheorem{prop}{Proposition}
\newtheorem{remark}{Remark}

\newtheorem{lemma}{Lemma}
\newtheorem{example}{Example}
\newtheorem{definition}{Definition}
\newtheorem{corollary}{Corollary}

\newtheorem{exercise}{Exercise}

\newcommand{\Ad}{\text{Ad}}

\newcommand{\jami}{Jamio{\l}kowski }

\newcommand{\cD}{\mathcal{D}}
\newcommand{\cE}{\mathcal{E}}

\newcommand{\cH}{\mathcal{H}}
\newcommand{\cI}{\mathcal{I}}
\newcommand{\cJ}{\mathcal{J}}

\newcommand{\cL}{\mathcal{L}}

\newcommand{\cW}{\mathcal{W}}

\newcommand{\bA}{\mathbb{A}}

\newcommand{\bC}{\mathbb{C}}

\newcommand{\bE}{\mathbb{E}}

\newcommand{\bK}{\mathbb{K}}

\newcommand{\bR}{\mathbb{R}}

\newcommand{\bU}{\mathbb{U}}
\newcommand{\bV}{\mathbb{V}}
\newcommand{\bW}{\mathbb{W}}

\begin{document}

\title{Quantum Metrology for Signals with an Unknown Causal Structure}

\author{Gwon Ryul Han}
\affiliation{Department of Physics, Ulsan National Institute of Science and Technology (UNIST), Ulsan 44919, Republic of Korea}

\author{Hyukgun Kwon}
\email{kwon37hg@sejong.ac.kr}
\affiliation{Department of Physics and Astronomy, Sejong University, 209 Neungdong-ro Gwangjin-gu, Seoul 05006, Republic of Korea}

\author{Seok Hyung Lie}
\email{seokhyung@unist.ac.kr}
\affiliation{Department of Physics, Ulsan National Institute of Science and Technology (UNIST), Ulsan 44919, Republic of Korea}

\begin{abstract}
Quantum sensing protocols usually presuppose where and when each signal acts, yet this information is unavailable beforehand in experiments such as scattering and radar. We show that causal agnosticity forces a measurement scheme to be interferometric even when only discrete interactions are allowed.
Describing the sensing events in the recently developed quantum state over spacetime formalism, we then prove that the Fisher information of an interferometer is bounded by the squared sum of the local generator norms, independently of how the events are ordered.
Nonunitary and spatiotemporally correlated signals can be explained as well when their physical dilation has a spacetime state on the enlarged events.
Notably, the bound holds without presupposing a background spacetime that orders the events whenever the full collection of target and auxiliary events has a spacetime state.
We show that this bound yields the spacetime Heisenberg limit and characterize its saturation through marginals of the spacetime state `dressed' by the signal. 
The $N^2$ scaling of the Fisher information with $N$ sensing events arises from the coherence of a single reference system shared across spacetime, whereas independently repeated interferometry obeys the standard quantum limit.
\end{abstract}

\maketitle

\paragraph{Introduction---}
 To achieve the ultimate precision in estimation of unknown parameters from quantum measurement data is one of the goals of quantum metrology. To this end, optimization of  probe states or dynamics through which the parameters are encoded has been studied~\cite{paris2009quantum,giovannetti2006quantum,giovannetti2011advances,pezze2018quantum}. This idea has been extended to networks of spatially~\cite{proctor2018multiparameter,zhang2021distributed,zhuang2018distributed} and temporally distributed sensors such as memory-assisted strategies using channels, combs, or process tensors~\cite{chiribella2009theoretical,pollock2018operational,milz2021quantum,altherr2021quantum}. Nevertheless, these theories often assume the capability of the experimenter either to know in advance or to manipulate when and where the signals are encoded. In other words, the causal structure of signals has been often assumed to be part of a sensing strategy.

However, some sensing schemes like quantum radar~\cite{lloyd2008enhanced,tan2008quantum} may involve probe--target scattering whose position and time, i.e., the causal structure of the signals, cannot be known in advance and must be inferred from the measurement outcomes. It is then desirable for the implementation of a sensing scheme to be independent of the causal structure of signals. This property is called \emph{causal agnosticity} and is known to force the corresponding probe--event interactions to be interferometric when they are continuous~\cite{lie2026probe}. Measurements outside this class require an additional spacetime reference frame, such as a clock or a map.

It was recently established that the configuration of quantum degrees of freedom distributed over spacetime and accessible through interferometry is characterized by a quantum state over spacetime~\cite{lie2026probe}, or briefly a \emph{spacetime state}~\cite{diaz2026unifying}, which generalizes density operators to arbitrary spacetime regions.
This observation gave operational meaning and a concrete experimental probing method for several related formalisms, including quantum states over time (QSOTs)~\cite{fullwood2022quantum,lie2024unique,lie2025multi}, pseudo-density operators (PDOs)~\cite{fitzsimons2015quantum,fullwood2025dynamics,marletto2021temporal}, and temporal entanglement constructions~\cite{milekhin2025observable}.

Our main result is a universal upper bound for an interferometric encoding in which signals distributed over spacetime act on one arm. For $N$ probe events at which unitary signals $U_{k,\theta}$ generated by $h_{k,\theta}$ act independently, we show the \emph{spacetime Heisenberg limit}
\begin{equation} \label{eqn:upper_bound_intro}
    F_\theta
    \leq
    \left(\sum_{k=1}^N \norm{h_{k,\theta}}_\infty\right)^2,
\end{equation}
where $\norm{\cdot}_\infty$ denotes the largest singular value.
We call Eq.~\eqref{eqn:upper_bound_intro} the spacetime Heisenberg limit because it scales as $O(N^2)$ without a prescribed event order. The bound follows solely from constraints on causally agnostic interferometric statistics, while its saturation is determined by marginals of the spacetime state associated with the sensing process.

\paragraph{Interferometry and spacetime states---}
We begin by fixing notation. All Hilbert spaces are assumed to be finite dimensional unless stated otherwise. We may drop system subscripts from operators when it is evident from the context. A pair $(\rho, \cE)$, consisting of an initial state $\rho_A$ and a channel $\cE : A \to B$, specifies a \emph{dynamics} from $A$ to $B$.
When a Hilbert space label refers to quantum degrees of freedom localized within one spacetime region rather than to a persistent physical carrier, we call it an \emph{event}. The label distinguishes the local region without presupposing its coordinates or causal relation to the other events. (See \hyperref[sec:end-systems-events]{\textit{Systems and events}} in the End Matter for our convention.)

Ordinary multipartite density operators yield joint measurement probabilities through the Born rule. At timelike separated events, earlier measurements can disturb later events, so their outcome effects alone do not specify joint statistics. Generalizations include higher-order objects, such as process matrices, process tensors, and superdensity operators~\cite{oreshkov2012quantum,chiribella2009theoretical,pollock2018operational,milz2021quantum,cotler2018superdensity}, and generalized state operators whose Born-rule pairings can only yield quasiprobabilities~\cite{leifer2013towards,gammelmark2013past,horsman2017can,fitzsimons2015quantum,fullwood2022quantum,lie2024unique,lie2025multi,fullwood2026quasiprobabilistic,lostaglio2023kirkwood}. For a state description on the event Hilbert spaces themselves, the question is how a possibly nonpositive or non-Hermitian operator can acquire a direct meaning in terms of observed probabilities.

Interferometry provides this connection by characterizing states through interference terms~\cite{lie2026probe}. A reference system $R$, realized as a spatial mode in a Mach-Zehnder interferometer or an internal degree of freedom in a Ramsey interferometer, is prepared in an equal superposition of the arms $\ket{0}_R$ and $\ket{1}_R$. An intervention unitary $V_A$ acts on the target $A$ prepared in $\rho_A$ only on arm $\ket{1}_R$. Measuring $R$ in the basis $\ket{\pm}_R$ after recombination gives
\begin{equation} \label{eqn:prob_1partite}
    P(\pm) = \frac{1\pm \Re\cI_\rho(V)}{2},
\end{equation}
where $\cI_\rho(V)=\Tr[V\rho]$ is the \emph{interference term} and $S_\rho(V):=\Re\cI_\rho(V)$ is the \emph{fringe}. Measuring in $\ket{\pm_\phi}_R=(\ket{0}_R\pm e^{i\phi}\ket{1}_R)/\sqrt{2}$ or applying $e^{-i\phi}V$ instead gives $P_\phi(\pm)=[1\pm S_\rho(e^{-i\phi}V)]/2$. Knowledge of both quadratures for arbitrary unitary $V_A$ uniquely determines $\rho$, so interference terms provide an equivalent characterization of the state through ordinary probabilities.

The crucial observation was that the same scheme can be extended to multiple events distributed over spacetime~\cite{lie2026probe} [Fig.~\ref{fig:scheme}(a)]. For $N$ events $\bA=A_1A_2\cdots A_N$ with an unspecified causal structure, the same reference system $R$ interacts with each $A_k$ through an intervention unitary $V^{(k)}_{A_k}$, collectively denoted by $\bV_\bA$, only on the arm $\ket{1}_R$. Operationally, the unitaries are chosen and applied locally at the corresponding spacetime events, which need not lie on a common time slice. The measurement probabilities $P(\pm)$ are still given by Eq.~\eqref{eqn:prob_1partite}, but now the interference term is generated by a single operator $\Omega$ on $\bA$, as derived in \hyperref[sec:end-interferometry]{\textit{Spacetime states from interferometry}} of the End Matter. When the events in $\bA$ are spacelike separated, $\Omega$ reduces to a conventional multipartite density operator $\rho_{A_1\cdots A_N}$. This motivates the following definition of \emph{spacetime states}.

\begin{definition}[Interferometric definition of spacetime states~\cite{lie2026probe}] \label{def:spacetime_states}
    We say that a collection of events $\bA=A_1A_2\cdots A_N$ is in a spacetime state $\Omega$ if the interference term of interferometry on $\bA$ with local intervention unitaries $V_{A_k}^{(k)}$ for $k=1,\dots, N$ is given as
    \begin{equation} \label{eqn:interference_term}
        \cI_\Omega(\bV_\bA)=\Tr[\qty(\bigotimes_{k=1}^N V^{(k)}_{A_k}) \Omega].
    \end{equation}
    Conversely, we say that an operator $\Omega$ on $\bA$ is a spacetime state if $\abs{\cI_\Omega(\bW_\bA)}\leq 1$ for every collection of local intervention unitaries $W^{(k)}_{A_k}$, represented by the tensor product $\bW_{\bA}=\bigotimes_{k=1}^N W^{(k)}_{A_k}$.
\end{definition}

The second clause may admit operators that no physical arrangement realizes. The upper bound proved below holds on this possibly larger set, whereas its saturation will be demonstrated with realizable states. Whenever a spacetime state is defined on timelike separated events, it is called a \emph{quantum state over time} (QSOT), or briefly a \emph{time state}. Even before this interferometric characterization, several works studied the correspondence between the conventional channel description of quantum dynamics and the time state description~\cite{fullwood2022quantum,lie2024unique,lie2025multi}. While the exact form of this correspondence is not crucial for understanding the main results of this work, we summarize the QSOT formalism in \hyperref[sec:end-time-states]{\textit{Time states and $\star$-products}} of the End Matter for self-containedness.

\paragraph{Characterizing interferometry---}%
Why is interferometry special in the spacetime state formalism? A general probe-based measurement scheme on events $\bA=A_1A_2\cdots A_N$ is specified by admissible unitary interactions between a reference system $R$ and the events, followed by a measurement of $R$. We call such a scheme \emph{causally agnostic} when its implementation requires no shared spacetime reference to coordinate the local interactions: for any independently chosen admissible interactions, the measurement process is invariant under the encounter order. This includes identical interactions at distinct events of the same type. Ref.~\cite{lie2026probe} showed that every causally agnostic scheme is an interferometer but assumed causal agnosticity even for infinitesimally short interactions between $R$ and each event $A_k$. We remove this assumption and allow the interactions to be discrete.

\begin{prop}\label{prop:causal_agnos_charac}
Any causally agnostic measurement scheme with discrete interactions can be implemented with an interferometer.
\end{prop}

The detailed setting and proof are given in Sec.~\ref{app:discrete_causal_agnostic} of the Supplemental Material (SM).

We focus on interferometry because the scheme itself is deterministically implementable, although weak measurements~\cite{aharonov1988result}, quantum snapshotting~\cite{lie2025multi,wang2024snapshotting}, and light-touch correlators~\cite{parzygnat2025time,liu2025quantum} access the same spacetime state information with a less direct readout~\cite{lie2026probe}. Since the pairwise fringes of a multi-arm interferometer can be sampled by statistically mixing two-arm settings, we henceforth use `interferometry' to mean two-arm interferometry unless stated otherwise.

\paragraph{Irregularity of spacetime states---}

If the reference system $R$ is prepared in $\ket{+}$ and the same process acts on both arms, omitting all interventions gives outcome $\ket{+}$ with unit probability. Since the corresponding interference term is $\cI_\Omega(I)=\Tr\Omega$, a freshly initialized apparatus implies $\Tr\Omega=1$. Reference noise or arm-dependent dynamics can instead change the surviving arm coherence and yield $\Tr\Omega\neq1$. A non-unit trace therefore records a change in the reference coherence even when the nominal interventions are omitted (see \hyperref[sec:end-irregularity]{\textit{Irregularity and dressing}} in the End Matter). We call a spacetime state $\Omega$ \emph{regular} if $\Tr\Omega=1$ and \emph{irregular} otherwise. The marginal of a regular spacetime state on each event is a density operator, as its trace norm is at most one by Definition~\ref{def:spacetime_states} while its trace equals one. This often overlooked class of irregular states will turn out to be useful for understanding causally agnostic metrology in the next section.

\begin{figure}
    \centering
    \includegraphics[width=\linewidth]{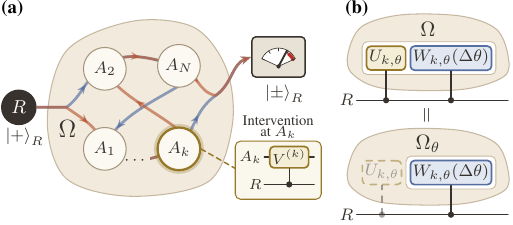}
    \caption{(a) Causally agnostic interferometry. A reference system $R$ prepared in $\ket{+}_R$ probes events in spacetime state $\Omega$ through local interventions $V^{(k)}_{A_k}$ controlled on $\ket{1}_R$, followed by measurement in $\ket{\pm}_R$. The same local probing prescription applies without specifying an encounter order. (b) Absorbing all fixed local signals $U_{k,\theta}$ into $\Omega_\theta=\bU_\theta\Omega$ leaves only the deviations $W_{k,\theta}(\Delta\theta)$ around $\theta$ as nominal interventions. The fixed signals act only on arm $\ket{1}_R$, so $\Omega_\theta$ can be irregular. }
    \label{fig:scheme}
\end{figure}

\paragraph{Metrology utilizing spacetime states---}%
We now study metrological schemes for signals with unspecified causal structure. Under the conditions of Proposition~\ref{prop:causal_agnos_charac}, a measurement that works for every encounter order has an interferometric implementation.
We therefore consider estimation of a parameter $\theta$ encoded in local unitary signals $U_{k,\theta}$ acting at the probe events $\bA=A_1A_2\cdots A_N$ with an unspecified causal structure. We call them `probe' events because the parameter $\theta$ will be encoded on them regardless of their causal structure.

In this two-arm setting, for a fixed spacetime state $\Omega$ that is independent of $\theta$ and a parameter-dependent intervention $\bU_\theta=\bigotimes_{k=1}^N U_{k,\theta}$ on one arm, the interferometer produces $P_\theta(\pm|\Omega)=(1\pm S_\theta(\Omega))/2$, where the fringe is $S_\theta(\Omega)=\Re\cI_\Omega(\bU_\theta)$, and hence the Fisher information of this measurement is
\begin{equation} \label{eqn:binary_FI}
    F_\theta (\Omega)
    =
    \frac{[\partial_\theta S_\theta(\Omega)]^2}{1-S_\theta(\Omega)^2}
\end{equation}
whenever $\abs{S_\theta(\Omega)}<1$. At points where $\abs{S_\theta(\Omega)}=1$, the Fisher information is understood by continuity whenever the corresponding limit exists.
The interferometric Fisher information has the following tight upper bound.

\begin{theorem}\label{thm:F_upper_bound}
Let $\bA = A_1A_2\cdots A_N$ be a collection of events in spacetime state $\Omega$.
For interferometry on $\bA$ with local unitary signals $U_{k,\theta}$ on $A_k$ satisfying
\begin{equation}
    \partial_\theta U_{k,\theta}=-ih_{k,\theta}U_{k,\theta},
\end{equation}
at any $\theta$ satisfying $\abs{S_\theta(\Omega)}<1$, the associated interferometric Fisher information $F_\theta(\Omega)$ obeys
\begin{equation} \label{eqn:upper_bound}
    F_\theta(\Omega) \leq \norm{\sum_{k=1}^N h_{k,\theta}}_\infty^2.
\end{equation}
\end{theorem}
Note that the tensor products of $h_{k,\theta}$ with identity operators on other events are suppressed in Eq.~\eqref{eqn:upper_bound} for readability.
When the preparation and all earlier dynamics preserve equal arm populations, as in the balanced realization used here, the output state of $R$ lies in the equatorial plane of the Bloch sphere. Its quantum Fisher information~\cite{helstrom1976quantum,braunstein1994statistical,paris2009quantum,giovannetti2006quantum} is attained by a projective readout in some equatorial basis $\qty{\ket{\pm_\phi}}$. This readout replaces $\Omega$ by $e^{-i\phi}\Omega$, to which Theorem~\ref{thm:F_upper_bound} applies unchanged. Thus Eq.~\eqref{eqn:upper_bound} also bounds the quantum Fisher information of $R$ in this balanced realization.

\paragraph{Universal Heisenberg limit---}%
From the subadditivity of the operator norm $\norm{\cdot}_\infty$, we immediately have the following result, in the same setting as Theorem~\ref{thm:F_upper_bound}.
\begin{corollary}[Spacetime Heisenberg limit] \label{coro:Heisenberg_limit}
\begin{equation} \label{eqn:HL}
    F_\theta(\Omega) \leq \left(\sum_{k=1}^N \norm{h_{k,\theta}}_\infty \right)^2 \leq \Lambda^2 N^2,
\end{equation}
where $\Lambda:=\max_k \norm{h_{k,\theta}}_\infty$.
\end{corollary}

An $O(N^2)$ Fisher information bound is the Heisenberg limit in standard parallel, sequential, and adaptive metrology~\cite{liu2023optimal,Kurdzialek2023Using}. Coherent control of order can lead to different precision limits in other metrological tasks~\cite{zhao2020metrology,mothe2024reassessing}. Equation~\eqref{eqn:HL}, by contrast, is stated without prescribing an encounter order. Its proof in the spacetime state formalism is elementary and given in Sec.~\ref{sec:ThmCoroProof} of the Supplemental Material.
Here $N$ counts the signal interactions on the intervention arm, while the shared reference is not an additional signal interaction. The operator norm appears rather than the spectral spread because a signal-arm phase is observable relative to the untouched reference arm. For $h_{k,\theta}=Z/2$ the bound is $N^2/4$ instead of the $N^2$ available when both probe levels see the signal.

We now characterize exactly when the first inequality in Eq.~\eqref{eqn:HL} is saturated for $\abs{S_\theta(\Omega)}<1$.
Since we estimate a local deviation $\Delta\theta$ about $\theta$, the fixed signal $\bU_\theta$ may be regarded as part of the probed process, albeit on one arm only, and absorbed into the possibly irregular `dressed' spacetime state
$\Omega_\theta:=\bU_\theta\Omega$ (Fig.~\ref{fig:scheme}(b); see \hyperref[sec:end-irregularity]{\textit{Irregularity and dressing}} in the End Matter).

Note that $\Omega_\theta$ still satisfies Definition~\ref{def:spacetime_states}. Let $\Omega_{k,\theta}$ be the `marginal' of $\Omega_\theta$ on $A_k$, i.e.,
\begin{equation}
    \Omega_{k,\theta}=\Tr_{\setminus k}[\bU_\theta \Omega],
\end{equation}
where $\Tr_{\setminus k}$ denotes the partial trace over all tensor factors except the $k$th one. It is an ordinary marginal of $\Omega_\theta$, the spacetime state of the events once the signal is counted as part of their dynamics conditioned on the reference arm. Note that $\Tr \Omega_{k,\theta} = \Tr \Omega_\theta=\cI_\Omega(\bU_\theta)$, so that $\Re \Tr \Omega_{k,\theta}=S_\theta(\Omega)$.
Equality in this bound holds if and only if the following two conditions are met (Remark~\ref{rem:saturation_conditions} of the Supplemental Material). First, every marginal state that encodes the parameter has maximal trace norm $\norm{\cdot}_1:=\Tr|\cdot|$, i.e.,
\begin{equation} \label{eqn:saturation_cond_1}
    \norm{\Omega_{k,\theta}}_1=1, \text{ whenever } h_{k,\theta}\neq 0.
\end{equation}
Second, each marginal is `aligned' in the sense that there exists a common sign $\sigma\in\{+1,-1\}$ such that
\begin{equation} \label{eqn:saturation_cond_2}
    \Im\Tr[h_{k,\theta}\Omega_{k,\theta}]
    =
    \sigma\norm{h_{k,\theta}}_\infty\sqrt{1-S_\theta(\Omega)^2},
\end{equation}
for all $k=1,\dots,N$.
Note that terms with $\norm{h_{k,\theta}}_\infty=0$ do not affect saturation. Reaching $N^2\Lambda^2$ additionally requires all local generator norms to equal $\Lambda$.

A convenient sufficient condition for saturation is maximal visibility, $\abs{\cI_\Omega(\bU_\theta)}=1$, with each dressed marginal $\Omega_{k,\theta}$ having its range in an extremal eigenspace of the generator $h_{k,\theta}$ with a common sign. That is, there exists $\sigma\in \qty{-1,+1}$ such that
\begin{equation}
    h_{k,\theta} \Omega_{k,\theta} = \sigma \norm{h_{k,\theta}}_\infty \Omega_{k,\theta},
\end{equation}
for all $k$. For one system visiting the events in a definite or classically random order, the undressed marginals give a sufficient condition. Suppose the system reaches every $A_k$ in a fixed occupied subspace contained in an extremal eigenspace of $h_{k,\theta}$ with a common sign throughout a parameter interval. If each signal acts as a scalar on that subspace at one reference point, Eq.~\eqref{eqn:HL} is saturated throughout the interval. The encounter order and parameter-independent noise between events may be arbitrary provided the dynamics deliver the system to these subspaces (Proposition~\ref{prop:directed_saturation} in the Supplemental Material). A common extremal sector preserved by the dynamics gives an order-agnostic realization even for noncommuting generators (Proposition~\ref{prop:order_agnostic}).

Notably, the bound \eqref{eqn:HL} can be saturated even when the spacetime state factorizes as $\Omega=\Omega_1 \otimes \Omega_2 \otimes \cdots \otimes \Omega_N$, provided that the conditions \eqref{eqn:saturation_cond_1} and \eqref{eqn:saturation_cond_2} are satisfied, hence the entanglement in spacetime state $\Omega$ is not essential. What is crucial for the saturation is the coherence of $R$ maintained over spacetime. For $N$ independent single-event interferometers, Fisher information is additive and scales only as $O(N)$, the standard quantum limit.

While a single coherent probe can attain the Heisenberg limit sequentially without entanglement~\cite{giovannetti2006quantum,higgins2007entanglement,vandam2007optimal}, Eq.~\eqref{eqn:HL} itself requires no prior specification of the encounter order or the dynamics between encounters. Saturation is ensured by the extremal-sector conditions above, which can be checked on the spacetime state marginals.

\paragraph{Beyond local unitary signals---}
The preceding discussion assumed local unitary signals at the probe events. More general physical encodings can be treated by adding auxiliary events $\bE=E_1E_2\cdots E_N$ in a spacetime state $\Xi_\bE$ and applying suitable unitaries $V_{k,\theta}$ on $A_kE_k$. As tensor products of process matrices with indefinite causal structure need not remain valid process matrices~\cite{jia2018tensor}, we use an auxiliary state only when $\Omega_\bA\otimes\Xi_\bE$ itself satisfies Definition~\ref{def:spacetime_states} for the enlarged events. The interference term is then
\begin{multline} \label{eqn:dilated_interference}
    \cI_{\Omega_\bA\otimes \Xi_\bE}(\bV_{\bA\bE}(\theta))
    =
    \Tr\left[
        \left(\bigotimes_{k=1}^N V_{k,\theta}\right)
        (\Omega_\bA \otimes \Xi_{\bE})
    \right].
\end{multline}
Under this condition, $\bigotimes_k V_{k,\theta}$ can be interpreted as an ordinary local unitary signal acting on each enlarged event $A_kE_k$. Open dynamics signals that are possibly nonunitary can be explained by discarding the auxiliary events, while correlations between signal events can be also supplied by the auxiliary spacetime state $\Xi_\bE$.
Note that the interference term in Eq.~\eqref{eqn:dilated_interference} can be written as $\Tr[\bK_\theta\Omega_\bA]$, where
\begin{equation}
    \bK_\theta
    =
    \Tr_\bE\left[
        \left(\bigotimes_{k=1}^N V_{k,\theta}\right)
        (I_\bA\otimes\Xi_\bE)
    \right].
\end{equation}
The \emph{intervention operator} $\bK_\theta$ fixes the interference statistics. Dilations of the same reduced channel can yield different $\bK_\theta$ and Fisher information.

For a local nonunitary signal, initialize each dilation auxiliary $E_k$ in $\ket{0}$. The inactive arm leaves it unchanged, selecting the Kraus amplitude $\bra{0}_{E_k}V_{k,\theta}\ket{0}_{E_k}$. With independent auxiliaries, $\bK_\theta$ is their product.
Spatially correlated signals represented by localizable maps~\cite{beckman2001causal}, i.e., multipartite channels implementable by local operations and a shared state without communication, are obtained by allowing $\Xi_\bE$ to be entangled. The resulting intervention operator $\bK_\theta$ need not factor into local Kraus operators.

Finally, the interferometric scheme can be causally agnostic even when the signal itself has temporal structure. If the encoding is a quantum process with a definite causal order, i.e., a quantum comb, it can be implemented by a memory system $E$ that sequentially interacts with the events through unitaries $V_{k,\theta}$~\cite{chiribella2008quantumcircuit,chiribella2009theoretical}. The dynamics of the memory can be represented by a time state $\Xi_\bE$ for the parameter-independent propagation of the memory between interactions, and the interference term again has the form of Eq.~\eqref{eqn:dilated_interference}. The relation between the intervention operator, the Kraus operators of a dilation, and correlated auxiliaries is detailed in \hyperref[sec:end-nonunitary-signals]{\emph{Nonunitary and correlated signals}} of the End Matter.

For a dilation with $\Omega_\bA\otimes\Xi_\bE$ satisfying Definition~\ref{def:spacetime_states}, Eq.~\eqref{eqn:HL} applies to the enlarged events and the dilation generators. Comparing compatible dilations with the same intervention operator $\bK_\theta$ can tighten the bound. The Supplemental Material gives explicit implementations for unknown scattering (Example~\ref{ex:unknown_event}), random collision patterns including independent collisions (Sec.~\ref{subsec:collision}), two implementations sharing a noisy channel but differing in Fisher information (Example~\ref{ex:two_impl}), and collectively controlled noise (Example~\ref{ex:collective}).

\paragraph{Discussion---}%
The quantum switch metrology studied in Refs.~\cite{zhao2020metrology,mothe2024reassessing} uses controlled schemes of the form~\eqref{eqn:controlled_form} in which both arms carry parameter-dependent signals. Absorbing the signal on the reference arm into the spacetime state yields a dressed state that depends on $\theta$, which lies outside the hypotheses of Theorem~\ref{thm:F_upper_bound}. These results are therefore not in tension with Eq.~\eqref{eqn:HL}. Their quantitative account within the spacetime state formalism is left for future work.

Within the hypotheses of Theorem~\ref{thm:F_upper_bound}, the saturation conditions~\eqref{eqn:saturation_cond_1} and \eqref{eqn:saturation_cond_2} can be met with echo-type control. Sign-flipping pulses inserted between the events select a signed combination of the local parameters, such as their sum or difference, and the encounter order within each interval between pulses is irrelevant (Proposition~\ref{prop:control_word}). For two events separated by one pulse and encountered in either order with probabilities $p$ and $1-p$, the in-phase fringe estimates their difference at the bound~\eqref{eqn:HL} for every $p$, while the quadrature fringe records the order imbalance $2p-1$ (Example~\ref{ex:echo}). Classically randomized photonic routing could test both quadratures, using routing elements from quantum switch experiments~\cite{procopio2015superposition,rubino2017experimental,goswami2018indefinite} with an order distribution independent of the interferometer arm.

In this work, we have developed a theory of causally agnostic metrology which Proposition~\ref{prop:causal_agnos_charac} characterizes as interferometric sensing for discrete probing interactions. Its Fisher information obeys a spacetime Heisenberg limit whose saturation is decided by marginals of the possibly irregular dressed spacetime state associated with the scheme. Theorem~\ref{thm:F_upper_bound} uses only the defining inequality $\abs{\cI_\Omega(\bW_\bA)}\leq1$, which holds because every fringe is a difference of probabilities. The bound therefore applies whenever the full collection of target and auxiliary events has a parameter-independent spacetime state and can be read out interferometrically. We remark that no background spacetime ordering enters the result. This makes the spacetime state formalism suitable for causally agnostic metrology of quantum events that cannot be embedded in a conventional spacetime such as indefinite causal order and retrocausality, and we leave demonstrating concrete applications to future work.

\begin{acknowledgments}
\paragraph{Acknowledgments---}%
S.H.L. is supported by the start-up fund and the 2026 Research Fund (1.250007.01, 1.260015.01) of Ulsan National Institute of Science \& Technology (UNIST), Institute of Information \& Communications Technology Planning \& Evaluation (IITP) Grants (RS-2025-02283189), and National Research Foundation of Korea (RS-2025-25464492, RS-2026-25520808). H.K. was supported by the IITP (RS-2025-25464252, RS-2024-00437191) and the NRF (RS-2025-25464492, RS-2024-00442710) funded by the Ministry of Science and ICT (MSIT), Korea.  We used OpenAI ChatGPT (Sol/Astra) and Anthropic Claude (Fable) for discussions, formatting and proofreading. The augmentation of the proof of Theorem~\ref{thm:F_upper_bound} using Bernstein's inequality was suggested by ChatGPT-5.6 Sol.
\end{acknowledgments}

\bibliography{main}

\clearpage
\setcounter{equation}{0}
\renewcommand{\theequation}{E\arabic{equation}}
\renewcommand{\theHequation}{E\arabic{equation}}
\onecolumngrid
\section*{End Matter}
\twocolumngrid
\phantomsection\label{sec:end-systems-events}
\paragraph{Systems and events---}%
We clarify the use of the terms `system' and `event' in this work. A \emph{system} denotes a physical carrier or platform whose identity may persist while it evolves in time. An \emph{event} denotes quantum degrees of freedom localized within one spacetime region and assigned their own Hilbert space label.

It should be noted that localization of an event in spacetime here need not be an assignment of specific spacetime coordinates $(x,t)$. Labels of events such as $A_1,A_2,\ldots$ distinguish local laboratories or encounters in the same way that labels like `Alice' and `Bob' in the conventional quantum theory distinguish separated laboratories, while their positions, times, and causal relations may remain unknown. In each physical realization, the events are localized, but the spacetime state description does not presuppose a background map that orders them.

The distinction is useful when one physical carrier appears at several encounters. For example, `the proton' is a system, whereas `the proton at encounter $A$' and `the proton at encounter $B$' are different events. However, this is only a convention for readability rather than an additional postulate, and the two terms may be used interchangeably when no confusion arises.

\phantomsection\label{sec:end-interferometry}
\paragraph{Spacetime states from interferometry---}%
The interferometric form of a causally agnostic scheme~\cite{lie2026probe} extends to discrete interactions after reference-only dynamics are absorbed into the preparation or readout of $R$ (Sec.~\ref{app:discrete_causal_agnostic} of the Supplemental Material). In a common arm basis, the local interactions take the controlled form
\begin{equation} \label{eqn:controlled_form}
    V_{RA_k}
    =
    \sum_\lambda
    \dyad{\lambda}_R\otimes V_{A_k|\lambda}.
\end{equation}
Here $V_{A_k|\lambda}$ denotes the unitary at event $A_k$ conditioned on arm $\lambda$. These interactions commute when represented on the labeled event spaces, so the same local probing prescription applies without a specified encounter order. This order independence concerns the implementation alone. The interference signal still depends on the process connecting the events, and this dependence is what the spacetime state records. In the two-arm scheme, we choose $V_{A_k|0}=I_{A_k}$ and $V_{A_k|1}=V^{(k)}_{A_k}$, keeping the original process on arm $\ket{0}_R$ as the reference for the interventions on arm $\ket{1}_R$.

At each encounter $A_k$, write the joint density operator $\varrho$ of $R$, the target, and any memory in blocks on $R$. The controlled interaction acts on its off-diagonal block as
\begin{equation}
    X_{10}:=\bra{1}_R\varrho\ket{0}_R
    \quad\longmapsto\quad
    (V^{(k)}_{A_k}\otimes I)X_{10},
    \label{eqn:coherence_insertion}
\end{equation}
where $I$ acts on the other degrees of freedom at that encounter. Evolution along the arms propagates this block linearly, including any memory. Thus each $V^{(k)}_{A_k}$ enters the final arm coherence once, by left multiplication at its own event. After tracing out the target and memory, let $\rho_R^{\mathrm{out}}(\bV_\bA)$ be the reference state at readout. Measuring $R$ in $\ket{\pm}$ gives $P(\pm)=[1\pm\Re\cI(\bV_\bA)]/2$, where the interference term is
\begin{equation}
    \cI(\bV_\bA)
    :=2\bra{1}\rho_R^{\mathrm{out}}(\bV_\bA)\ket{0}.
    \label{eqn:coherence_readout}
\end{equation}

With the preparation and propagation fixed, $\cI$ is multilinear in the local intervention matrices. In bases $\{\ket{i_k}\}$ of the event spaces, it therefore admits an expansion
\begin{equation}
    \cI(\bV_\bA)
    =\sum_{\boldsymbol i,\boldsymbol j}
    c_{\boldsymbol i\boldsymbol j}
    \prod_k(V^{(k)}_{A_k})_{i_kj_k}.
    \label{eqn:interference_coefficients}
\end{equation}
These coefficients contain the preparation, dynamics, and causal connections between the events. With $\ket{\boldsymbol i}:=\bigotimes_k\ket{i_k}$, define
\begin{equation}
    \Omega
    :=\sum_{\boldsymbol i,\boldsymbol j}
    c_{\boldsymbol i\boldsymbol j}
    \dyad{\boldsymbol j}{\boldsymbol i}.
    \label{eqn:state_from_interference}
\end{equation}
Then $\cI(\bV_\bA)=\Tr[\bV_\bA\Omega]$, recovering Eq.~\eqref{eqn:interference_term}. Each tensor factor of $\Omega$ refers to an intervention event, including successive encounters of one carrier, rather than to a subsystem of the final output. Since local unitaries span each operator space, this representation is unique. Positivity of $\rho_R^{\mathrm{out}}$ gives $\abs{\cI}\leq1$, so $\Omega$ satisfies Definition~\ref{def:spacetime_states}. Combined with Proposition~\ref{prop:causal_agnos_charac}, this shows that every causally agnostic scheme probes a spacetime state of the events, whatever their causal relations.

\phantomsection\label{sec:end-irregularity}
\paragraph{Irregularity and dressing---}%
With every intervention set to the identity, $\Tr\Omega=2\bra{1}\rho_R^{\mathrm{out}}(I)\ket{0}$ measures the arm coherence. For example, starting from $\dyad{+}_R\otimes\rho$, unitary evolutions $U_0,U_1$ on the respective arms give $X_{10}=U_1\rho U_0^\dag/2$ and hence $\Tr\Omega=\Tr[U_1\rho U_0^\dag]$, whereas identical evolutions on both arms give one. Reference dephasing attenuates $X_{10}$ and hence $\Tr\Omega$, while $\rho_R^{\mathrm{out}}$ remains normalized. Earlier dynamics can also restore the initial coherence, so regularity does not exclude prior arm-dependent evolution.

In particular, a fixed unitary $U^{(k)}_{A_k}$ applied before $V^{(k)}_{A_k}$ on arm $\ket{1}_R$ at each event can be interpreted as an arm-dependent evolution. Regarding these fixed unitaries, collectively denoted by $\bU_\bA$, as part of the probed process changes $\Omega$ to the `dressed' state $\bU_\bA\Omega$, which the subsequent interventions $\bV_\bA$ probe through
\begin{equation}
    \cI_\Omega(\bV_\bA\bU_\bA)
    =
    \cI_{\bU_\bA\Omega}(\bV_\bA).
    \label{eqn:absorb_intervention}
\end{equation}
We remark that since $\bV_\bA\bU_\bA$ is again a local unitary intervention, $\bU_\bA\Omega$ satisfies Definition~\ref{def:spacetime_states}. Its trace $\Tr[\bU_\bA\Omega]=\cI_\Omega(\bU_\bA)$ is the interference term of the fixed intervention itself, so the dressed state is irregular unless this term equals one.

In local estimation, this freedom in dividing interventions into the nominal intervention part and the arm-dependent evolution part allows the signal at a fixed operating point $\theta$ to be absorbed into the probed process. Define the remaining deviation by
\begin{equation}
    W_{k,\theta}(\Delta\theta)
    :=
    U_{k,\theta+\Delta\theta}U_{k,\theta}^\dag
    =
    I-ih_{k,\theta}\Delta\theta+O(\Delta\theta^2).
\end{equation}
With $\bW_\theta(\Delta\theta)=\bigotimes_kW_{k,\theta}(\Delta\theta)$, the signal factorizes as $\bU_{\theta+\Delta\theta}=\bW_\theta(\Delta\theta)\bU_\theta$. Equation~\eqref{eqn:absorb_intervention} with $\bU_\bA=\bU_\theta$ and $\bV_\bA=\bW_\theta(\Delta\theta)$ then identifies $\Omega_\theta=\bU_\theta\Omega$ as the dressed state probed by $\bW_\theta(\Delta\theta)$ [Fig.~\ref{fig:scheme}(b)]. Here $\theta$ is held fixed, so $\Omega_\theta$ is independent of the estimated deviation $\Delta\theta$, while it is generally irregular. To first order, the deviation at $A_k$ contributes $\Im\Tr[h_{k,\theta}\Omega_{k,\theta}]$ to the fringe slope, where $\Omega_{k,\theta}$ is the marginal of $\Omega_\theta$ on $A_k$. These local sensitivities explain why the saturation conditions~\eqref{eqn:saturation_cond_1} and \eqref{eqn:saturation_cond_2} involve the dressed marginals.

\phantomsection\label{sec:end-time-states}
\paragraph{Time states and $\star$-products---}%
For self-containedness, we summarize the construction of time states using the operation known as the $\star$-product in this section.
For a concrete two-event example, let $\rho_A$ evolve to $B$ through a channel $\cE$ common to both arms. With a fresh reference, the coherence block undergoes left multiplication by $V_A$, propagation by $\cE$, and left multiplication by $V_B$, giving
\begin{equation}
    \cI_\Omega(V_A\otimes V_B)
    =\Tr[V_B\cE(V_A\rho_A)].
    \label{eqn:two_event_interference}
\end{equation}
To identify the time state of this dynamics, we rewrite the interference term as a trace pairing on $A\otimes B$. With the \jami operator~\cite{jamiolkowski1972linear,choi1975completely}
\begin{equation}
    J[\cE]
    =\sum_{i,j}\dyad{i}{j}_A\otimes\cE(\dyad{j}{i}),
\end{equation}
we obtain
\begin{equation}
    \Tr[V_B\cE(V_A\rho_A)]     =\Tr[(V_A\otimes V_B)(\rho_A\otimes I_B)J[\cE]].
\end{equation}
Since this equality holds for every pair of local intervention unitaries, comparison with Eq.~\eqref{eqn:interference_term} identifies the time state probed by this scheme as an object known as the left $\star$-product~\cite{fullwood2022quantum,lie2024unique,lie2025multi},
\begin{equation}
    \Omega=\cE\star_L\rho=(\rho_A\otimes I_B)J[\cE].
\end{equation}
Repeating this construction along a channel chain gives the directed time states used in Sec.~\ref{sec:chronology} of the Supplemental Material.

For comparison, the right and Fullwood-Parzygnat (FP) products~\cite{fullwood2022quantum,lie2024unique} are defined by
\begin{align}
    \cE\star_R\rho&=J[\cE](\rho_A\otimes I_B),
    \\
    \cE\star_{\mathrm{FP}}\rho
    &=\tfrac12(\cE\star_L\rho+\cE\star_R\rho).
\end{align}
The right product replaces $V_A\rho_A$ in Eq.~\eqref{eqn:two_event_interference} by $\rho_AV_A$, while the FP product averages the two orderings. The left and right pairings are complex conjugates for Hermitian observables and give the Kirkwood-Dirac quasiprobabilities when evaluated on spectral projectors~\cite{kirkwood1933quantum,dirac1945analogy,lostaglio2023kirkwood}. Their real part gives the Margenau-Hill distribution~\cite{margenau1961correlation,lostaglio2023kirkwood}. The generally non-Hermitian left and right products signify time reversal asymmetry, whereas the FP product is Hermitian. Here time reversal asymmetry concerns distinguishability of future and past, rather than the reversibility of the channel~\cite{lie2026probe}. The two-time FP product coincides with the pseudo-density operator, and its canonical iteration yields the multipartite version for a channel chain~\cite{fitzsimons2015quantum,fullwood2022quantum,fullwood2025dynamics,parzygnat2025time,diaz2026unifying}.

\phantomsection\label{sec:end-nonunitary-signals}
\paragraph{Nonunitary and correlated signals---}%
The main text described how a dilated signal is absorbed into the intervention operator $\bK_\theta$. Here, we explain the technical information retained by $\bK_\theta$ that is absent from the reduced channel.

For a local dilation with the auxiliary event initialized in $\ket{0}_{E_k}$, define
\begin{equation}
    K_{k,\theta}^{(\mu)}
    :=
    \bra{\mu}_{E_k}V_{k,\theta}\ket{0}_{E_k}.
\end{equation}
While the reduced signal channel induced by the Kraus operators is
\begin{equation}
    \cE_{k,\theta}(\tau)
    =
    \sum_\mu
    K_{k,\theta}^{(\mu)}\tau
    K_{k,\theta}^{(\mu)\dag},
\end{equation}
since the inactive interferometer arm leaves $E_k$ in $\ket{0}_{E_k}$, only the reference-overlap amplitude $K_{k,\theta}^{(0)}$ is relevant in the interferometry~\cite{oi2003interference,abbott2020communication}. For independently prepared auxiliaries,
\begin{equation}
    \bK_\theta
    =
    \bigotimes_{k=1}^N K_{k,\theta}^{(0)}.
\end{equation}

Correlations in the auxiliary spacetime state supply nonlocal structure in the intervention operator. For any product operator expansion
\begin{equation}
    \Xi_\bE
    =
    \sum_\alpha c_\alpha
    \bigotimes_{k=1}^N X_k^{(\alpha)},
\end{equation}
one obtains
\begin{align}
    \bK_\theta
    &=
    \sum_\alpha c_\alpha
    \bigotimes_{k=1}^N K_{k,\theta}^{(\alpha)},
    \\
    K_{k,\theta}^{(\alpha)}
    &:=
    \Tr_{E_k}\left[
        V_{k,\theta}
        \left(I_{A_k}\otimes X_k^{(\alpha)}\right)
    \right].
\end{align}
An uncorrelated auxiliary gives a product intervention operator, while a correlated auxiliary can give a nonfactorizing one although every $V_{k,\theta}$ is local on $A_kE_k$. The same formula covers a definite-order memory process when the $E_k$ label successive events of one memory system and $\Xi_\bE$ is its time state.

Finally, if $\partial_\theta V_{k,\theta}=-i\widetilde h_{k,\theta}V_{k,\theta}$, the Heisenberg limit applies to the enlarged event partition:
\begin{equation}
    F_\theta
    \leq
    \norm{\sum_{k=1}^N\widetilde h_{k,\theta}}_\infty^2
    \leq
    \left(
        \sum_{k=1}^N\norm{\widetilde h_{k,\theta}}_\infty
    \right)^2.
\end{equation}
 This bound is obtained via a specific dilation, whereas $\bK_\theta$ is the observable object. Distinct compatible dilations may share the same $\bK_\theta$ and therefore the same statistics in interferometry while giving different dilation-dependent bounds of the Fisher information.

% References cited only in the Supplemental Material are retained in the main bibliography.
\nocite{redheffer1953plancherel,bisio2011quantum}

\clearpage
% Restore the numbered sections of the standalone PRA supplement after the PRL main text.
\setcounter{secnumdepth}{4}
\makeatletter
\renewcommand{\theHequation}{S\arabic{equation}}
\renewcommand{\theHsection}{S\arabic{section}}
\renewcommand{\theHsubsection}{S\arabic{section}.\arabic{subsection}}
\providecommand{\theHtheorem}{}
\providecommand{\theHlemma}{}
\providecommand{\theHcorollary}{}
\providecommand{\theHprop}{}
\providecommand{\theHexample}{}
\providecommand{\theHremark}{}
\renewcommand{\theHtheorem}{S\arabic{theorem}}
\renewcommand{\theHlemma}{S\arabic{lemma}}
\renewcommand{\theHcorollary}{S\arabic{corollary}}
\renewcommand{\theHprop}{S\arabic{prop}}
\renewcommand{\theHexample}{S\arabic{example}}
\renewcommand{\theHremark}{S\arabic{remark}}
\makeatother
\onecolumngrid
\begin{center}
    {\bf \large Supplemental Material for ``Quantum Metrology for Signals with an Unknown Causal Structure''}
\end{center}

% --- Independent numbering for the Supplemental Material -------------
\setcounter{section}{0}
\setcounter{equation}{0}
\setcounter{theorem}{0}
\setcounter{lemma}{0}
\setcounter{corollary}{0}
\setcounter{prop}{0}
\setcounter{example}{0}
\setcounter{remark}{0}
\renewcommand{\thesection}{S\arabic{section}}
\renewcommand{\theequation}{S\arabic{equation}}
\renewcommand{\thetheorem}{S\arabic{theorem}}
\renewcommand{\thelemma}{S\arabic{lemma}}
\renewcommand{\thecorollary}{S\arabic{corollary}}
\renewcommand{\theprop}{S\arabic{prop}}
\renewcommand{\theexample}{S\arabic{example}}
\renewcommand{\theremark}{S\arabic{remark}}

\section{Conventions}
For a system or event $A$, we write $\cH_A$ for its Hilbert space, and $\cL(\cH)$ denotes the space of linear operators on a Hilbert space $\cH$. The parameter is encoded at $N$ probe events, labeled $A_1,A_2,\ldots,A_N$ and written collectively as $\bA=A_1A_2\cdots A_N$. These labels distinguish the events without fixing their causal relations, so the enumeration implies no encounter order. Boldface denotes collectives of this kind throughout, both for events and for the joint operators assembled from their single-event constituents. In particular, $E_k$ denotes an auxiliary event attached to $A_k$, whether it dilates a nonunitary signal (Sec.~\ref{sec:beyond}) or carries memory of the probe dynamics (Sec.~\ref{sec:chronology}), and $\bE=E_1E_2\cdots E_N$. Products and sums over $k$ run from $1$ to $N$ unless stated otherwise. Whenever the encoding consists of local unitary signals, we write the joint signal as $\bU_\theta = \bigotimes_k U_{k,\theta}$. Their Hermitian generators $h_{k,\theta}$ are defined by $\partial_\theta U_{k,\theta}=-ih_{k,\theta}U_{k,\theta}$, with $h_{k,\theta}=h_{k,\theta}^\dag$. The complex quantity
\begin{equation}
    \cI_\Omega(\bU_\theta) = \Tr[\bU_\theta\Omega],
    \qquad\text{abbreviated}\qquad
    \cI_\Omega(\theta) := \cI_\Omega(\bU_\theta),
\end{equation}
is the \emph{interference term}; it carries the spacetime state $\Omega$ on which it is evaluated as a subscript. The measured signal, the \emph{fringe}, is its real part, and a controllable readout phase $\phi$ on the reference system gives the family
\begin{equation}
    S^{(\phi)}_\theta(\Omega) = \Re\qty[e^{-i\phi}\,\cI_\Omega(\theta)],
    \qquad
    S_\theta(\Omega) := S^{(0)}_\theta(\Omega) = \Re\,\cI_\Omega(\theta),
    \qquad
    S^{(\pi/2)}_\theta(\Omega) = \Im\,\cI_\Omega(\theta).
\end{equation}
We call $S_\theta$ the (in-phase) fringe and $S^{(\pi/2)}_\theta$ the \emph{out-of-phase} fringe. Subscripts on $S$ and on the Fisher information $F$ denote the estimated parameter; the readout phase appears only as the superscript $(\phi)$, omitted when it is zero. At points where $\abs{S_\theta}=1$, the Fisher information is understood by the limiting convention stated in the main text.

\section{Discrete characterization of causally agnostic interactions}
\label{app:discrete_causal_agnostic}

We prove a discrete version of the characterization of interferometry as a causally agnostic measurement scheme~\cite{lie2026probe}. The continuous-time argument of Ref.~\cite{lie2026probe} imposed order independence on infinitesimal evolutions generated by Hamiltonians. Here we impose order independence directly on finitely many discrete interactions between a reference system and target systems.

Let $R$ be the reference system and let $A$ denote a target event. For a unitary $V_{RA}$, write the corresponding unitary channel as $\Ad_V(\cdot)=V(\cdot)V^\dag$.

A reference-based measurement scheme is specified by its admissible interaction sets. Let $\mathfrak U_A$ be the admissible family between $R$ and an event of type $A$, independently available at every copy of that event type (in particular, the same member may be reused at distinct copies). We impose causal agnosticity on every finite sequence of event types, including repetitions. It requires interactions chosen independently from these families to define the same measurement for every encounter order, whatever target process is connected to the events, while the final outcome is read only from $R$.

To compare complete measurement processes rather than a single terminal marginal, it is useful to employ the quantum comb formalism~\cite{chiribella2008quantumcircuit,chiribella2009theoretical,bisio2011quantum}. A comb represents a sequential quantum network by one Choi operator, retaining its open input and output ports, while a tester represents an experiment connected to those ports. An operation called the link product that connects a comb with a tester gives the outcome probabilities. This language is particularly convenient here because different encounter orders can be compared as operators once their external ports are identified by the event labels.

Fix a specific choice of admissible interactions $\mathbf V=(V_{RA_1},\ldots,V_{RA_N})$. For every permutation $\pi\in S_N$, their ordered application defines a pure deterministic comb $\mathsf C_\pi[\mathbf V]$ whose memory is $R$ and whose teeth encounter $A_{\pi(1)},\ldots,A_{\pi(N)}$. We label its external ports by the corresponding events and, for simplicity, use the same label for input and output ports. After the ports are identified in this way, all these combs have the same input-output type. Let $C_\pi[\mathbf V]$ denote the Choi operator of $\mathsf C_\pi[\mathbf{V}]$.

In this picture, the tester plays the role of the target process and of the experimenter at once, in the sense that it supplies whatever enters the event ports and collects whatever leaves them, and it also prepares $R$ and performs the final measurement on $R$. Not every tester of the right port type can be connected to $\mathsf C_\pi[\mathbf V]$, however. Since the teeth of $\mathsf C_\pi[\mathbf V]$ are applied in the order $\pi$, a tester that routes the output of an event encountered later in $\pi$ into the input of an event encountered earlier would close a causal loop. We say that a tester is \emph{compatible} with the encounter order $\pi$ when it creates no such loop, that is, when its own causal structure can be interleaved with that of $\mathsf C_\pi[\mathbf V]$. Causal agnosticity compares the statistics of one and the same tester across different encounter orders, so only testers compatible with every $\pi\in S_N$ enter the condition. Such testers exist in abundance: a tester that prepares an arbitrary state of $R\bA$ (possibly correlated with an ancillary memory) in a common past, delivers each $A_k$ to its input port, retains the output of every event port in quantum memory, processes all retained systems in a common future, and finally measures $R$ alone never feeds one event port into another, and is therefore compatible with every encounter order, including those in which the events are timelike separated. The testers used below are all of this form.

For a tester $\mathsf T=\{T_x\}_x$ compatible with every encounter order, whose recorded output is the result $x$ of the final measurement on $R$, the generalized Born rule for quantum networks gives~\cite{chiribella2009theoretical}
\begin{equation}
    p_\pi(x|\mathbf V,\mathsf T)
    =
    T_x*C_\pi[\mathbf V],
    \label{eqn:comb_born_rule}
\end{equation}
where $*$ denotes the link product mentioned above. When the encounter order changes, the same tester is connected to the correspondingly labeled ports. Causal agnosticity requires Eq.~\eqref{eqn:comb_born_rule} to be independent of $\pi$ for every admissible $\mathbf V$, every tester $\mathsf T$ compatible with every encounter order, and every outcome $x$.

\begin{lemma}
\label{lem:statistics_imply_channel}
Under the causal-agnostic condition above, the ordered interaction combs coincide,
\begin{equation}
    C_\pi[\mathbf V]
    =
    C_{\pi'}[\mathbf V]
    \label{eqn:ordered_comb_equality}
\end{equation}
for every pair of permutations $\pi,\pi'$. In particular, any two independently chosen admissible interactions at distinct event copies commute as unitary channels. In other words, for any admissible interactions $V_{RA}$ and $W_{RB}$,
\begin{equation}
    \Ad_{V_{RA}}\Ad_{W_{RB}}
    =
    \Ad_{W_{RB}}\Ad_{V_{RA}}.
    \label{eqn:channel_level_commutation}
\end{equation}
\end{lemma}

\begin{proof}
Composing the teeth of $\mathsf C_\pi[\mathbf V]$ in their encounter order gives the unitary channel
\begin{equation}
    \Phi_\pi[\mathbf V]
    :=
    \Ad_{V_{RA_{\pi(N)}}}
    \cdots
    \Ad_{V_{RA_{\pi(1)}}}
    \label{eqn:ordered_total_channel}
\end{equation}
from $R\bA$ to $R\bA$. Once the ports are identified by the event labels, $C_\pi[\mathbf V]$ is the Choi operator of $\Phi_\pi[\mathbf V]$, up to a fixed reordering of the tensor factors that is the same for every $\pi$. Equation~\eqref{eqn:ordered_comb_equality} is therefore equivalent to $\Phi_\pi[\mathbf V]=\Phi_{\pi'}[\mathbf V]$. Two different combs can always be told apart by some tester~\cite{chiribella2009theoretical}; what has to be checked is that a tester compatible with every encounter order, whose only measurement acts on $R$, already suffices.

Suppose that $\Phi_\pi[\mathbf V]\neq\Phi_{\pi'}[\mathbf V]$. Since product density operators span the operator space of $R\bA$, there is a state $\rho_R\otimes\rho_\bA$ for which
\begin{equation}
    \Phi_\pi[\mathbf V](\rho_R\otimes\rho_\bA)
    \neq
    \Phi_{\pi'}[\mathbf V](\rho_R\otimes\rho_\bA),
\end{equation}
and the two output states are separated by an effect $0\leq E\leq I$ on $R\bA$, $\Tr[E\,\Phi_\pi[\mathbf V](\rho_R\otimes\rho_\bA)]\neq\Tr[E\,\Phi_{\pi'}[\mathbf V](\rho_R\otimes\rho_\bA)]$. Consider the tester that prepares $\rho_R\otimes\rho_\bA$ in the common past, retains the output of every event port, applies in the common future the channel
\begin{equation}
    \cD_E(\tau)
    =
    \Tr[E\tau]\dyad{1}_R
    +
    \Tr[(I-E)\tau]\dyad{0}_R
    \label{eqn:readout_channel_to_R}
\end{equation}
from $R\bA$ to $R$, and measures $R$ in the basis $\{\ket{0}_R,\ket{1}_R\}$. This tester is of the common-past--common-future form described above, so it is compatible with every encounter order, and its only measurement acts on $R$. Connected to $\mathsf C_\pi[\mathbf V]$, it returns $x=1$ with probability $\Tr[E\,\Phi_\pi[\mathbf V](\rho_R\otimes\rho_\bA)]$, which differs between $\pi$ and $\pi'$. This contradicts causal agnosticity. Therefore all ordered channels, and equivalently all ordered combs, coincide, proving Eq.~\eqref{eqn:ordered_comb_equality}.

It remains to extract pairwise commutativity. Choose two permutations that differ only by exchanging adjacent encounters $A_i$ and $A_j$. Their total channels can be written as
\begin{align}
    \Phi_\pi
    &=
    \Phi_>\,
    \Ad_{V_{RA_j}}
    \Ad_{V_{RA_i}}\,
    \Phi_<,
    \\
    \Phi_{\pi'}
    &=
    \Phi_>\,
    \Ad_{V_{RA_i}}
    \Ad_{V_{RA_j}}\,
    \Phi_<,
\end{align}
where $\Phi_<$ and $\Phi_>$ are the common channels before and after the exchanged pair. Since they are unitary channels and hence invertible, canceling them from $\Phi_\pi=\Phi_{\pi'}$ gives
\begin{equation}
    \Ad_{V_{RA_j}}\Ad_{V_{RA_i}}
    =
    \Ad_{V_{RA_i}}\Ad_{V_{RA_j}}.
\end{equation}
Since the interactions at the two copies can be chosen independently, setting $V_{RA_i}=V_{RA}$ and $V_{RA_j}=W_{RB}$ gives Eq.~\eqref{eqn:channel_level_commutation}.
\end{proof}

Equation~\eqref{eqn:channel_level_commutation} thus follows at the channel level from requiring identical statistics against every tester compatible with every encounter order. Equality of a single terminal marginal of $R$ would not suffice, but equality of the complete $R$-readout statistics does.
To turn this channel-level relation into a controlled operator form, we first remove the possible projective phase when the same interaction is placed at two event copies.

\begin{lemma}
\label{lem:no_projective_phase_same_copy}
Let $A$ and $B$ be two copies of the same event, and let $V_{RA}$ and $V_{RB}$ be the same admissible unitary applied to the two copies. 
If
\begin{equation}
    \Ad_{V_{RA}}\Ad_{V_{RB}}
    =
    \Ad_{V_{RB}}\Ad_{V_{RA}},
\end{equation}
then
\begin{equation}
    V_{RA}V_{RB}
    =
    V_{RB}V_{RA}.
\end{equation}
\end{lemma}

\begin{proof}
The equality of unitary channels implies that
\begin{equation}
    V_{RA}V_{RB}
    =
    \omega
    V_{RB}V_{RA}
    \label{eqn:self_projective_commutation}
\end{equation}
for some phase $\omega$. (One can see this from, for example, their Choi matrices.)
Let $S_{AB}$ be the swap between the two event copies. 
Since
\begin{equation}
    S_{AB}V_{RA}S_{AB}=V_{RB},
    \qquad
    S_{AB}V_{RB}S_{AB}=V_{RA},
\end{equation}
Eq.~\eqref{eqn:self_projective_commutation} gives, with $X=V_{RA}V_{RB}$,
\begin{equation}
    S_{AB}XS_{AB}
    =
    \omega^{-1}X .
\end{equation}
The channel $\Ad_S$, conjugation by $S_{AB}$, has only the eigenvalues $\pm1$, so $\omega=\pm1$. 
If $\omega=-1$, then $S_{AB}X=-XS_{AB}$. 
Since $X$ is unitary, it would map the symmetric subspace of $AB$ bijectively onto the antisymmetric subspace. 
This is impossible because their dimensions are
\begin{equation}
    (\dim\cH_R) \times \frac{d_A(d_A+1)}{2}
    \quad\text{and}\quad
    (\dim\cH_R) \times \frac{d_A(d_A-1)}{2},
\end{equation}
which are unequal. 
Therefore $\omega=1$.
\end{proof}

Together with Lemma~\ref{lem:statistics_imply_channel}, this removes the projective ambiguity and gives exact self-commutation for every admissible interaction. The next lemma converts that relation into a controlled form for each interaction separately.

\begin{lemma}
\label{lem:self_commuting_is_controlled}
For any bipartite unitary operator $V_{RA}$, if
\begin{equation}
    V_{RA}V_{RB}
    =
    V_{RB}V_{RA}
\end{equation}
for a second copy $B$ of $A$, then there is an orthogonal decomposition
\begin{equation}
    I_R=\sum_\lambda P_\lambda
\end{equation}
such that
\begin{equation}
    V_{RA}
    =
    \sum_\lambda P_\lambda\otimes V_\lambda^A,
    \label{eqn:controlled_form_single}
\end{equation}
where every $V_\lambda^A$ is unitary on $A$.
\end{lemma}

\begin{proof}
Let
\begin{equation}
    V_{RA}
    =
    \sum_a X_a^R\otimes F_a^A
    \label{eqn:schmidt_expansion_V}
\end{equation}
be an operator Schmidt expansion with $\{F_a^A\}_a$ linearly independent. 
Substituting Eq.~\eqref{eqn:schmidt_expansion_V} into the commutation relation gives
\begin{equation}
    \sum_{a,b}
    [X_a,X_b]\otimes F_a^A\otimes F_b^B
    =
    0 .
\end{equation}
Since $\{F_a^A\otimes F_b^B\}_{a,b}$ is linearly independent because of the linear independence of $\{F_a^A\}_a$,
\begin{equation}
    [X_a,X_b]=0
\end{equation}
for all $a,b$. 
Thus the matrices $\{X_a\}_a$ are simultaneously upper triangularizable with respect to some basis of $R$. In that basis, the unitary $V_{RA}$ is block upper triangular and hence block diagonal. The linear independence of $\{F_a^A\}_a$ then implies that every $X_a$ is diagonal. Grouping equal joint eigenvalues, we obtain
\begin{equation}
    V_{RA}
    =
    \sum_\lambda P_\lambda\otimes V_\lambda^A .
\end{equation}
The unitarity of $V_{RA}$ implies
\begin{equation}
    (V_\lambda^A)^\dag V_\lambda^A
    =
    V_\lambda^A(V_\lambda^A)^\dag
    =
    I_A
\end{equation}
for every $\lambda$.
\end{proof}

Thus each admissible interaction is individually controlled, but its reference decomposition may still depend on the interaction. The next lemma uses the pairwise channel commutativity from Lemma~\ref{lem:statistics_imply_channel} to align the decompositions of two interactions, allowing reference-only dynamics within their common blocks.

\begin{lemma}
\label{lem:common_block_form}
Let $V_{RA}$ and $W_{RB}$ be two admissible unitary interactions of a causally agnostic measurement scheme
satisfying Eq.~\eqref{eqn:channel_level_commutation}. Then there is a common orthogonal decomposition
\begin{equation}
    I_R
    =
    \sum_\lambda\Pi_\lambda
\end{equation}
such that
\begin{equation}
    V_{RA}
    =
    \sum_\lambda
    X_\lambda\otimes V_\lambda^A,
    \qquad
    W_{RB}
    =
    \sum_\lambda
    Y_\lambda\otimes W_\lambda^B,
    \label{eqn:common_block_form}
\end{equation}
where $V_\lambda^A$ and $W_\lambda^B$ are unitary, while
$X_\lambda$ and $Y_\lambda$ are unitary on the support of
$\Pi_\lambda$:
\begin{equation}
    X_\lambda^\dag X_\lambda
    =
    X_\lambda X_\lambda^\dag
    =
    Y_\lambda^\dag Y_\lambda
    =
    Y_\lambda Y_\lambda^\dag
    =
    \Pi_\lambda.
\end{equation}
Equivalently, there are unitaries $X_R^V$ and $Y_R^W$ that are block
diagonal with respect to $\{\Pi_\lambda\}_\lambda$ such that
\begin{align}
    V_{RA}
    &=
    \left(
        X_R^V\otimes I_A
    \right)
    \sum_\lambda
    \Pi_\lambda\otimes V_\lambda^A,
    \\
    W_{RB}
    &=
    \left(
        Y_R^W\otimes I_B
    \right)
    \sum_\lambda
    \Pi_\lambda\otimes W_\lambda^B.
\end{align}
\end{lemma}

\begin{proof}
By Lemma~\ref{lem:self_commuting_is_controlled}, there exist projectors $\qty{P_i}$ and $\qty{Q_j}$ on $R$ that allow the following expressions,
\begin{equation}
    V_{RA}
    =
    \sum_iP_i\otimes V_i^A,
    \qquad
    W_{RB}
    =
    \sum_jQ_j\otimes W_j^B.
    \label{eqn:individual_controlled_forms}
\end{equation}
Group the indices of $V$ according to the following equivalence relation $\sim_V$,
\begin{equation}
    i\sim_V k
    \quad\Longleftrightarrow\quad
    V_i^A\propto V_k^A.
\end{equation}
Note that two unitaries $V_i$ and $V_k$ with $i\sim_V k$ can only differ by a phase factor. 
For each equivalence class $\alpha$ induced by $\sim_V$, choose a unitary
$\widehat V_\alpha^A$ and phases $v_i$ such that
\begin{equation}
    V_i^A
    =
    v_i\widehat V_\alpha^A
\end{equation}
for every $i\in\alpha$. Define
\begin{equation}
    P_\alpha
    :=
    \sum_{i\in\alpha}P_i,
    \qquad
    X_\alpha
    :=
    \sum_{i\in\alpha}v_iP_i.
\end{equation}
Then, since $|v_i|^2=1$ for every $i$,
\begin{equation}
    V_{RA}
    =
    \sum_\alpha
    X_\alpha\otimes\widehat V_\alpha^A,
    \qquad
    X_\alpha^\dag X_\alpha
    =
    X_\alpha X_\alpha^\dag
    =
    P_\alpha.
    \label{eqn:phase_reduced_V}
\end{equation}
Define $Q_\beta$, $Y_\beta$, and $\widehat W_\beta^B$ analogously for $W$ with similarly defined equivalence classes $\beta$,
so that
\begin{equation}
    W_{RB}
    =
    \sum_\beta
    Y_\beta\otimes\widehat W_\beta^B.
    \label{eqn:phase_reduced_W}
\end{equation}

Equation~\eqref{eqn:channel_level_commutation} implies
\begin{equation}
    V_{RA}W_{RB}
    -
    \omega
    W_{RB}V_{RA}
    =
    0
    \label{eqn:pair_projective_commutation}
\end{equation}
for some phase $\omega$. Taking the $(i,k)$ block with respect to
$\{P_i\}_i$ gives
\begin{equation}
    \left(
        V_i^A-\omega V_k^A
    \right)
    \otimes
    P_iW_{RB}P_k
    =
    0.
    \label{eqn:V_class_block}
\end{equation}
Thus a nonzero block $P_iW_{RB}P_k$ is possible only when
$i\sim_V k$. It follows that
\begin{equation}
    [W_{RB},P_\alpha]=0
\end{equation}
for every equivalence class $\alpha$. The same argument with $V$ and $W$ exchanged
shows that
\begin{equation}
    [V_{RA},Q_\beta]=0
\end{equation}
for every $\beta$.

Taking the $(i,k)$ block of the latter equality gives
\begin{equation}
    \left(
        V_i^A-V_k^A
    \right)
    \otimes
    P_iQ_\beta P_k
    =
    0.
    \label{eqn:class_projector_block}
\end{equation}
If $i$ and $k$ belong to distinct equivalence classes, their
event-side unitaries are not proportional and hence are not equal.
Therefore
\begin{equation}
    P_\alpha Q_\beta P_{\alpha'}
    =
    0
\end{equation}
whenever $\alpha\neq\alpha'$. Consequently, by adding all the $\alpha'$ except $\alpha$, one can show that
\begin{equation}
    [P_\alpha,Q_\beta]=0
\end{equation}
for every $\alpha,\beta$.

Now we obtain finer projectors through the following procedure. Discarding zero products, define
\begin{equation}
    \Pi_{\alpha\beta}
    :=
    P_\alpha Q_\beta.
\end{equation}
These are mutually orthogonal projectors and satisfy
\begin{equation}
    I_R
    =
    \sum_{\alpha,\beta}
    \Pi_{\alpha\beta}.
\end{equation}
Moreover, Eqs.~\eqref{eqn:phase_reduced_V} and
\eqref{eqn:phase_reduced_W}, together with
$[V_{RA},Q_\beta]=[W_{RB},P_\alpha]=0$, imply
\begin{equation}
    [X_\alpha,Q_\beta]=0,
    \qquad
    [Y_\beta,P_\alpha]=0.
\end{equation}
Indeed, $[V_{RA},Q_\beta]=\sum_\alpha[X_\alpha,Q_\beta]\otimes\widehat V_\alpha^A$, and since $X_\alpha$ and $Q_\beta$ both commute with $P_\alpha$, the term with index $\alpha$ is supported on the block $P_\alpha$ of $R$; the terms with distinct $\alpha$ are therefore supported on mutually orthogonal blocks and cannot cancel, so each vanishes separately (no linear independence of the $\widehat V_\alpha^A$ is needed).
Hence
\begin{equation}
    X_{\alpha\beta}
    :=
    \Pi_{\alpha\beta}X_\alpha,
    \qquad
    Y_{\alpha\beta}
    :=
    \Pi_{\alpha\beta}Y_\beta
\end{equation}
are unitary on the support of $\Pi_{\alpha\beta}$. We finally obtain
\begin{align}
    V_{RA}
    &=
    \sum_{\alpha,\beta}
    X_{\alpha\beta}\otimes\widehat V_\alpha^A,
    \\
    W_{RB}
    &=
    \sum_{\alpha,\beta}
    Y_{\alpha\beta}\otimes\widehat W_\beta^B.
\end{align}
Relabeling the nonzero pairs $(\alpha,\beta)$ by $\lambda$ proves
Eq.~\eqref{eqn:common_block_form}.
\end{proof}

Lemma~\ref{lem:common_block_form} supplies a common refinement for each pair and the proof of Theorem~\ref{thm:discrete_normal_form} extends the same block argument to the entire admissible family. It then collects the residual reference-only factors to obtain the multi-arm interferometric form underlying Proposition~\ref{prop:causal_agnos_charac} in the main text.
\begin{theorem}
\label{thm:discrete_normal_form}
Under causal agnosticity as defined above, there is a common orthogonal decomposition of the reference system,
\begin{equation}
    I_R
    =
    \sum_\eta \Pi_\eta,
\end{equation}
such that every admissible interaction $V_{RA}$ is controlled by this decomposition up to a reference-only unitary. More precisely,
\begin{equation}
    \begin{aligned}
        \widetilde V_{RA}
        &:=
        \sum_\eta \Pi_\eta\otimes V_\eta^A,
        \\
        V_{RA}
        &=
        \left(X_R^V\otimes I_A\right)\widetilde V_{RA}
        =
        \widetilde V_{RA}\left(X_R^V\otimes I_A\right),
    \end{aligned}
    \label{eqn:common_controlled_form}
\end{equation}
where $V_\eta^A$ and $X_R^V$ are unitary and
$[X_R^V,\Pi_\eta]=0$ for every $\eta$.
The reference-only factors in any chosen sequence of admissible interactions can be collected before or after the controlled interactions and absorbed into the preparation or readout of $R$. Hence the resulting measurement scheme is operationally a multi-arm interferometer with arms $\{\Pi_\eta\}_\eta$.
\end{theorem}

\begin{proof}
For each admissible interaction $V_{RA}$, choose the phase-reduced decomposition used in the proof of Lemma~\ref{lem:common_block_form}:
\begin{equation}
    V_{RA}
    =
    \left(X_R^V\otimes I_A\right)
    \sum_\alpha P_\alpha^V\otimes \widehat V_\alpha^A,
    \label{eqn:phase_reduced_family_form}
\end{equation}
where the event-side unitaries $\widehat V_\alpha^A$ belonging to distinct sectors are not proportional, $X_R^V$ is unitary, and $[X_R^V,P_\alpha^V]=0$ for every $\alpha$. For any two admissible interactions $V$ and $W$, the block argument in Lemma~\ref{lem:common_block_form} gives
\begin{equation}
    [P_\alpha^V,P_\beta^W]
    =
    [X_R^V,P_\beta^W]
    =
    0
    \label{eqn:family_sector_commutation}
\end{equation}
for all $\alpha$ and $\beta$.
Thus all reduced sector projectors generate a finite-dimensional commutative algebra on $R$. Let $\{\Pi_\eta\}_\eta$ be the minimal projections of this algebra. They refine every $\{P_\alpha^V\}_\alpha$, and Eq.~\eqref{eqn:family_sector_commutation} implies
\begin{equation}
    [X_R^V,\Pi_\eta]=0
\end{equation}
for every admissible $V$ and every $\eta$. Refining Eq.~\eqref{eqn:phase_reduced_family_form} by $\{\Pi_\eta\}_\eta$ gives Eq.~\eqref{eqn:common_controlled_form}. The commutation relation in that equation also shows that the reference-only factor may be placed on either side of the controlled interaction.

It remains to verify that these factors do not change the interferometric implementation. For a chosen collection of interactions at distinct event copies, write
$V_k=(X_k\otimes I_{A_k})\widetilde V_k$ as in Eq.~\eqref{eqn:common_controlled_form}, with identity factors on the remaining events suppressed. Since every $X_k$ commutes with all $\Pi_\eta$, it commutes with every $\widetilde V_l$. The controlled interactions also commute with one another because they use the same projectors on $R$ and act on distinct event systems. Therefore, for any encounter order $\pi$,
\begin{equation}
    V_{\pi(N)}\cdots V_{\pi(1)}
    =
    X_{\pi(N)}\cdots X_{\pi(1)}
    \widetilde V_N\cdots\widetilde V_1.
    \label{eqn:collect_reference_dynamics}
\end{equation}
The second factor is the order-independent multi-arm controlled interaction. Causal agnosticity of the full channel, together with the invertibility of this controlled factor, then implies
\begin{equation}
    \Ad_{X_{\pi(N)}\cdots X_{\pi(1)}}
    =
    \Ad_{X_{\pi'(N)}\cdots X_{\pi'(1)}}
    \label{eqn:reference_channel_order_independence}
\end{equation}
for all $\pi$ and $\pi'$. Hence the accumulated reference dynamics is independent of the encounter order as a channel. It can be absorbed into the initial reference state or, equivalently, into the final POVM by conjugating its effects. After this redefinition, only the controlled factors $\widetilde V_k$ remain, yielding precisely a multi-arm interferometric implementation.

\end{proof}

\section{Proof of Theorem~\ref{thm:F_upper_bound} and Corollary~\ref{coro:Heisenberg_limit}} \label{sec:ThmCoroProof}

We first prove the global bound using a pointwise form of Bernstein's inequality, then derive the spacetime Heisenberg limit and its saturation conditions from the dressed marginals. A phase-rate bound $\abs{s'}\leq\Lambda\sqrt{1-s^2}$ bounds the rate of change of $\arccos s$ by $\Lambda$ wherever $\abs{s}<1$; for the fringe $s=S_\theta(\Omega)$, the binary Fisher information is the square of this rate.

\begin{lemma}[Phase-rate form of Bernstein's inequality]\label{Slem:Bernstein_endpoint}
Let $f(t)=\sum_{j=1}^m c_j e^{i\omega_jt}$ be real-valued for real $t$, with $c_j\in\bC$ and $\omega_j\in[-\Lambda,\Lambda]$, where $\Lambda\geq0$. If $\abs{f(t)}\leq1$ for all $t\in\bR$, then
\begin{equation}
    \abs{f'(t)}\leq\Lambda\sqrt{1-f(t)^2}
\end{equation}
for all $t\in\bR$.
\end{lemma}

\begin{proof}
The case $\Lambda=0$ is immediate. Otherwise, the entire extension satisfies $\abs{f(x+iy)}\leq(\sum_j\abs{c_j})e^{\Lambda\abs{y}}$, so its exponential type is at most $\Lambda$. The pointwise Bernstein inequality of Ref.~\cite{redheffer1953plancherel} states that a real entire function $g$ of exponential type strictly less than $\alpha$, with $\abs{g(t)}<1$ on the real axis, obeys $g(t)^2+g'(t)^2/\alpha^2\leq1$. Apply it to $g=f/(1+\delta)$ and $\alpha=\Lambda+\epsilon$, with $\delta,\epsilon>0$, to obtain
\begin{equation}
    f(t)^2+\frac{f'(t)^2}{(\Lambda+\epsilon)^2}\leq(1+\delta)^2.
\end{equation}
Letting $\epsilon,\delta\to0^+$ proves the claim, including the endpoints $f(t)=\pm1$.
\end{proof}

\begin{theorem}[Phase-rate bound for spacetime states]\label{Sthm:F_upper_bound}
Let $\Omega\in\cL(\bigotimes_{k=1}^N\cH_k)$ be a spacetime state, so that $\abs{\Tr[(\bigotimes_k V_k)\Omega]}\leq1$ for every collection of local unitaries. Let $\bU_\theta=\bigotimes_k U_{k,\theta}$ be differentiable, with $\partial_\theta U_{k,\theta}=-ih_{k,\theta}U_{k,\theta}$ and $h_{k,\theta}=h_{k,\theta}^\dag$, and set $S_\theta(\Omega)=\Re\Tr[\bU_\theta\Omega]$. Then
\begin{equation}
    \abs{\partial_\theta S_\theta(\Omega)}
    \leq
    \norm{\sum_k h_{k,\theta}^{[k]}}_\infty
    \sqrt{1-S_\theta(\Omega)^2},
\end{equation}
where $h_{k,\theta}^{[k]}=I_1\otimes\cdots\otimes I_{k-1}\otimes h_{k,\theta}\otimes I_{k+1}\otimes\cdots\otimes I_N$. Consequently, the binary probabilities $P_\theta(\pm)=[1\pm S_\theta(\Omega)]/2$ have Fisher information
\begin{equation}
    F_\theta(\Omega)
    =\frac{(\partial_\theta S_\theta(\Omega))^2}{1-S_\theta(\Omega)^2}
    \leq\norm{\sum_k h_{k,\theta}^{[k]}}_\infty^2
\end{equation}
whenever $\abs{S_\theta(\Omega)}<1$.
\end{theorem}

\begin{proof}
Fix $\theta$, set $G_\theta=\sum_k h_{k,\theta}^{[k]}$ and $\Lambda_\theta=\norm{G_\theta}_\infty$, and define
\begin{equation}
    f_\theta(t)=\Re\Tr[e^{-itG_\theta}\bU_\theta\Omega].
\end{equation}
Since the lifted generators commute,
\begin{equation}
    e^{-itG_\theta}\bU_\theta
    =\bigotimes_k(e^{-ith_{k,\theta}}U_{k,\theta}),
\end{equation}
so the definition of a spacetime state gives $\abs{f_\theta(t)}\leq1$ for all real $t$. This real exponential polynomial has frequencies $\pm\lambda$, with $\lambda$ an eigenvalue of $G_\theta$, all lying in $[-\Lambda_\theta,\Lambda_\theta]$. Moreover, $\partial_\theta\bU_\theta=-iG_\theta\bU_\theta$ implies
\begin{equation}
    f_\theta(0)=S_\theta(\Omega),
    \qquad
    f_\theta'(0)=\Re\Tr[-iG_\theta\bU_\theta\Omega]=\partial_\theta S_\theta(\Omega).
\end{equation}
Lemma~\ref{Slem:Bernstein_endpoint} at $t=0$ therefore proves the phase-rate bound, including $\Lambda_\theta=0$. Dividing its square by $1-S_\theta(\Omega)^2$ gives the Fisher information bound in the interior.
\end{proof}

The same Bernstein inequality gives the trace norm estimate needed to identify saturation of the spacetime Heisenberg limit.

\begin{lemma}[Trace norm phase-rate bound]\label{Slem:Trace_norm_endpoint}
Let $\Omega\in\cL(\cH)$ be arbitrary and let $\bU_\theta$ be any differentiable family of unitaries on $\cH$, not necessarily of product form, with Hermitian generator $G_\theta$ defined by $\partial_\theta\bU_\theta=-iG_\theta\bU_\theta$. For $S_\theta(\Omega)=\Re\Tr[\bU_\theta\Omega]$,
\begin{equation}
    \abs{\partial_\theta S_\theta(\Omega)}
    \leq\norm{G_\theta}_\infty\sqrt{\norm{\Omega}_1^2-S_\theta(\Omega)^2}.
\end{equation}
\end{lemma}

\begin{proof}
The case $\norm{\Omega}_1=0$ is immediate. Otherwise fix $\theta$ and apply Lemma~\ref{Slem:Bernstein_endpoint} to
\begin{equation}
    f_\theta(t)=\frac{\Re\Tr[e^{-itG_\theta}\bU_\theta\Omega]}{\norm{\Omega}_1}.
\end{equation}
Trace norm duality gives $\abs{f_\theta(t)}\leq1$, and its frequencies lie in $[-\norm{G_\theta}_\infty,\norm{G_\theta}_\infty]$. Since $f_\theta(0)=S_\theta(\Omega)/\norm{\Omega}_1$ and $f_\theta'(0)=\partial_\theta S_\theta(\Omega)/\norm{\Omega}_1$, the result follows.
\end{proof}

\begin{corollary}\label{Scoro:HL}
In the setting of Theorem~\ref{Sthm:F_upper_bound}, define
\begin{equation}\label{eqn:Lambda_tot}
    \Lambda_{\rm tot}:=\sum_k\norm{h_{k,\theta}}_\infty.
\end{equation}
Then
\begin{equation}
    \abs{\partial_\theta S_\theta(\Omega)}\leq\Lambda_{\rm tot}\sqrt{1-S_\theta(\Omega)^2},
    \qquad
    F_\theta\leq\Lambda_{\rm tot}^2,
\end{equation}
where the Fisher information bound applies whenever $\abs{S_\theta(\Omega)}<1$.
\end{corollary}

\begin{proof}
The bound follows from Theorem~\ref{Sthm:F_upper_bound} by norm subadditivity. To obtain its equality conditions, define the dressed marginals $\Omega_{k,\theta}=\Tr_{\setminus k}[\bU_\theta\Omega]$, where $\Tr_{\setminus k}$ traces out all tensor factors except $k$. They satisfy $\Re\Tr\Omega_{k,\theta}=S_\theta(\Omega)$ and
\begin{equation}
    \norm{\Omega_{k,\theta}}_1
    =\max_{V_k\in\mathrm{U}(\cH_k)}\abs{\Tr[V_k\Omega_{k,\theta}]}
    =\max_{V_k\in\mathrm{U}(\cH_k)}\abs{\Tr[V_k^{[k]}\bU_\theta\Omega]}\leq1,
\end{equation}
because $V_k^{[k]}\bU_\theta$ is a product of local unitaries. Differentiating the fringe and applying Lemma~\ref{Slem:Trace_norm_endpoint} to $\Re\Tr[e^{-ith_{k,\theta}}\Omega_{k,\theta}]$ at $t=0$ gives the chain
\begin{align}
    \abs{\partial_\theta S_\theta(\Omega)}
    &=\abs{\sum_k\Im\Tr[h_{k,\theta}\Omega_{k,\theta}]}\nonumber\\
    &\leq\sum_k\abs{\Im\Tr[h_{k,\theta}\Omega_{k,\theta}]}\label{eqn:sum_of_abs}\\
    &\leq\sum_k\norm{h_{k,\theta}}_\infty\sqrt{\norm{\Omega_{k,\theta}}_1^2-S_\theta(\Omega)^2}\label{eqn:marginal_phase_rate}\\
    &\leq\Lambda_{\rm tot}\sqrt{1-S_\theta(\Omega)^2}.\label{eqn:saturation_chain}
\end{align}
This proves the corollary and supplies the saturation conditions below.
\end{proof}

The corollary extends to $\abs{S_\theta(\Omega)}=1$ by continuity whenever the relevant limits exist. The first derivative vanishes at these endpoints, so the limiting Fisher information depends on higher order behavior, hence the following saturation conditions concern interior points only.

\begin{remark}\label{rem:saturation_conditions}
Fix $\theta$ with $\abs{S_\theta(\Omega)}<1$, and call event $k$ \emph{active} if $\norm{h_{k,\theta}}_\infty>0$. Equality in Eq.~\eqref{eqn:sum_of_abs} requires a common sign for its nonzero terms. Equality in Eq.~\eqref{eqn:marginal_phase_rate} requires saturation of the trace norm phase-rate bound at each active event, while equality in Eq.~\eqref{eqn:saturation_chain} requires $\norm{\Omega_{k,\theta}}_1=1$ there. Thus the spacetime Heisenberg limit is saturated if and only if, for every active $k$,
\begin{equation}
    \norm{\Omega_{k,\theta}}_1=1,
    \qquad
    \Im\Tr[h_{k,\theta}\Omega_{k,\theta}]
    =\sigma\norm{h_{k,\theta}}_\infty\sqrt{1-S_\theta(\Omega)^2}
\end{equation}
with a common sign $\sigma\in\{+1,-1\}$. Note that inactive events impose no condition and if all events are inactive, we have $F_\theta=\Lambda_{\rm tot}^2=0$.

We can acquire useful sufficient conditions in the following special case. For some common $\sigma\in\{+1,-1\}$,
\begin{equation}
    h_{k,\theta}\Omega_{k,\theta}
    =\sigma\norm{h_{k,\theta}}_\infty\Omega_{k,\theta}
\end{equation}
for every active $k$. Since $\Tr\Omega_{k,\theta}=\cI_\Omega(\bU_\theta)$, substitution into the fringe derivative gives
\begin{equation}
    \partial_\theta S_\theta(\Omega)=\sigma\Lambda_{\rm tot}\Im\cI_\Omega(\bU_\theta),
    \qquad
    F_\theta=\Lambda_{\rm tot}^2\frac{(\Im\cI_\Omega(\bU_\theta))^2}{1-(\Re\cI_\Omega(\bU_\theta))^2}.
\end{equation}
The spacetime state bound $\abs{\cI_\Omega(\bU_\theta)}\leq1$ then implies $F_\theta\leq\Lambda_{\rm tot}^2$. If at least one event is active, equality holds precisely at maximal visibility, $\abs{\cI_\Omega(\bU_\theta)}=1$.
\end{remark}

\section{Examples of signal encodings beyond local unitaries}
\label{sec:beyond}

Throughout this section, the implementation of signals is given in the form of a dilation $\bV_\theta=\bigotimes_{k=1}^N V_{k,\theta}$, written $\bV_{\bA\bE}(\theta)$ in the main text, in which $V_{k,\theta}$ acts on the enlarged event $A_kE_k$, together with a parameter-independent auxiliary spacetime state $\Xi_\bE$. As in the main text, the interference term is governed by the \emph{intervention operator}
\begin{equation} \label{eqn:intervention}
    \bK_\theta = \Tr_\bE\qty[\bV_\theta\,(I_\bA\otimes\Xi_\bE)],
    \qquad
    \cI_{\Omega_\bA\otimes\Xi_\bE}(\bV_\theta) = \Tr[\bK_\theta\Omega_\bA] .
\end{equation}
We consider only choices for which $\Omega_\bA\otimes\Xi_\bE$ satisfies the spacetime state condition for the enlarged events $A_kE_k$.
The two sides of the second equality express two viewpoints on the same experiment. On the left, the auxiliary events belong to the probed system and the signal is a local unitary on $\bA\bE$, so that the theory of the preceding sections applies verbatim. On the right, only $\bA$ is regarded as the target, and the auxiliaries have been absorbed into the signal, which is now represented by $\bK_\theta$ and need be neither unitary nor local. We adopt the second viewpoint in this section and write $\cI_\Omega(\theta)=\Tr[\bK_\theta\Omega]$ for a probe $\Omega$ of the target events. Since $\bK_\theta$, and with it the fringe, depends on the dilation and not only on the reduced signal channel, each example specifies its dilation explicitly. Section~\ref{subsec:collision} treats unitary signals that act at a random subset of the events, which is the situation of a probe scattering off targets at unknown positions or times, and Sec.~\ref{subsec:implementation} compares two dilations of one noisy channel that differ in their interferometric Fisher information.

\subsection{Random collision patterns}
\label{subsec:collision}

In the main text, we only considered the signals deterministically acting on certain events, but probabilistic signals can be described in our framework by introducing an appropriate auxiliary event as well. Here, we begin with the simplest case of a single active event at an unknown location.

\begin{example}
\label{ex:unknown_event}
Let a single scatterer imprint $U_\theta=e^{-i\theta Z/2}$ at exactly one event $A_k$, with a $\theta$-independent probability $P(k)$, and act as the identity at the others. This is the radar-type situation described in the main text, in which the probe meets one target whose position among $N$ sites, or whose time of encounter among $N$ windows, is not known in advance. The pattern is implemented by classical registers $E_k$ in the correlated state $\Xi_\bE=\sum_kP(k)\dyad{\boldsymbol e_k}$, where $\boldsymbol e_k$ is the bit string with a single $1$ at position $k$, and by the controlled unitaries $V_{k,\theta}=I\otimes\dyad{0}_{E_k}+U_\theta\otimes\dyad{1}_{E_k}$, so that Eq.~\eqref{eqn:intervention} gives the convex sum of local unitaries $\bK_\theta=\sum_k P(k)U_\theta^{[k]}$, where $U_\theta^{[k]}$ is the tensor product of $U_\theta$ on $A_k$ with identities elsewhere. Writing $\Omega_k=\Tr_{\setminus k}\Omega$ for the marginal at the $k$th event of a probe spacetime state $\Omega_\bA$,
\begin{equation}
    \cI_\Omega(\theta)
    =
    \Tr\qty[\qty(\sum_k P(k)U_\theta^{[k]})\Omega]
    =
    \Tr\qty[U_\theta\qty(\sum_k P(k)\Omega_k)] .
\end{equation}
Take the product probe $\Omega=\dyad{0}^{\otimes N}$, whose local vector is aligned with the maximal eigenvalue, $(Z/2)\ket{0}=\Lambda\ket{0}$, where $\Lambda=\norm{Z/2}_\infty=1/2$. 
All marginals equal $\dyad{0}$ independently of $k$, so $\sum_k P(k)\Omega_k=\dyad{0}$ and
\begin{equation}
    \cI_\Omega(\theta) = \mel{0}{U_\theta}{0} = e^{-i\theta/2} .
\end{equation}
Hence $S_\theta=\cos(\theta/2)$ and $F_\theta=\Lambda^2$ for every $\theta$. The last expression shows that the statistics coincide with those of single event interferometry with the signal $U_\theta$ on the averaged marginal $\sum_kP(k)\Omega_k$, to which Theorem~\ref{thm:F_upper_bound} applies with $N=1$. The unknown location therefore causes no loss relative to the single-signal bound because the relevant marginals are identical. The $N$ candidate sites and the common reference remain implementation resources. Note that the auxiliary state $\Xi_\bE$ can be replaced by the coherent superposition of single-collision patterns $\sum_{kl} \sqrt{P(k)P(l)} \dyad{\boldsymbol e _k}{\boldsymbol e_l}$ and still yield the same result.
\end{example}

We can generalize this observation to arbitrary probabilistic collision patterns as follows.

\begin{prop}
\label{prop:char_func}
Let $\boldsymbol{s}=(s_1,\ldots,s_N)\in\{0,1\}^N$ denote a collision pattern, where $s_k=1$ means that $U_\theta=e^{-i\theta h}$ acts at event $A_k$, and let $P(\boldsymbol{s})$ be a $\theta$-independent distribution. Implement the pattern by classical registers in the state $\Xi_\bE=\sum_{\boldsymbol s}P(\boldsymbol s)\dyad{\boldsymbol s}_\bE$ and by the controlled unitaries $V_{k,\theta}=I\otimes\dyad{0}_{E_k}+U_\theta\otimes\dyad{1}_{E_k}$. Then
\begin{equation}
    \bK_\theta
    =
    \sum_{\boldsymbol{s}\in\{0,1\}^N}
    P(\boldsymbol{s})
    \bigotimes_{k=1}^N U_\theta^{s_k}.
    \label{eqn:collision_pattern_K}
\end{equation}
For the aligned probe $\Omega=\dyad{v}^{\otimes N}$ with $h\ket{v}=\sigma\Lambda\ket{v}$ and $\Lambda=\norm{h}_\infty$,
\begin{equation}
\begin{aligned}
    \cI_\Omega(\theta)
    &=
    \Tr[\bK_\theta\Omega]
    \\
    &=
    \sum_{\boldsymbol{s}} P(\boldsymbol{s})e^{-i\sigma n(\boldsymbol{s})\Lambda\theta}
    \\
    &=
    \sum_{n=0}^N P(n)e^{-i\sigma n\Lambda\theta}
    =
    \expval{e^{-i\sigma n\Lambda\theta}},
\end{aligned}
\end{equation}
where $n(\boldsymbol{s})=\sum_k s_k$ is the number of active events, the collision number, and $P(n)=\sum_{\boldsymbol{s}:n(\boldsymbol{s})=n}P(\boldsymbol{s})$. Thus the interference term is the characteristic function of the collision number. The fringe $S_\theta=\expval{\cos(n\Lambda\theta)}$ obeys
\begin{equation}
    S_\theta
    =
    1-\frac{\Lambda^2\theta^2}{2}\expval{n^2}+O(\theta^4),
    \qquad
    \partial_\theta S_\theta
    =
    -\Lambda^2\theta\expval{n^2}+O(\theta^3),
\end{equation}
and the limiting Fisher information as $\theta\to0$ is
\begin{equation}
    F_{\theta\to0}
    =
    \Lambda^2\expval{n^2}.
\end{equation}
Moreover, at every operating point with $\abs{S_\theta}<1$,
\begin{equation}
    F_\theta
    \leq
    \Lambda^2\expval{n^2}.
    \label{eqn:collision_bound}
\end{equation}
Indeed, $\partial_\theta S_\theta=-\Lambda\expval{n\sin(n\Lambda\theta)}$, so the Cauchy--Schwarz inequality $\expval{n\sin(n\Lambda\theta)}^2\leq\expval{n^2}\expval{\sin^2(n\Lambda\theta)}$ together with $\expval{\cos^2(n\Lambda\theta)}\geq\expval{\cos(n\Lambda\theta)}^2$ gives $(\partial_\theta S_\theta)^2\leq\Lambda^2\expval{n^2}(1-S_\theta^2)$.
\end{prop}

The interferometric Fisher information of a random collision pattern is thus governed by the second moment of the collision number. Equation~\eqref{eqn:collision_bound} is the spacetime Heisenberg limit of Corollary~\ref{coro:Heisenberg_limit} with $N$ replaced by the root-mean-square collision number $\expval{n^2}^{1/2}$, and it is attained at $\theta\to0$. Since $n\leq N$, it never exceeds $N^2\Lambda^2$, which is also the bound that \emph{Nonunitary and correlated signals} of the End Matter assigns to this dilation, whose generators $h\otimes\dyad{1}_{E_k}$ have norm $\Lambda$.

For independent collisions with a known, $\theta$-independent probability $q$ at each event, as in a dilute sample, the collision number is binomial. Proposition~\ref{prop:char_func} then gives
\begin{equation}
    F_{\theta\to0}
    =
    \Lambda^2\qty[(qN)^2+q(1-q)N].
\end{equation}
This retains $O(N^2)$ scaling at fixed $q>0$ and reduces to $N^2\Lambda^2$ at $q=1$.

\subsection{Implementation dependence of noisy signals}
\label{subsec:implementation}

Having treated random occurrences of unitary signals, we now compare two physical implementations of the same noisy channel, which induce different intervention operators and hence different interferometric Fisher information.

\begin{example}
\label{ex:two_impl}
Consider the phase-flip channel composed with the signal $U_\theta=e^{-i\theta Z/2}$,
\begin{equation}
    \cE_\theta(\rho)
    =
    (1-p)\,U_\theta\rho U_\theta^\dag
    + p\,Z U_\theta\rho U_\theta^\dag Z ,
\end{equation}
and evaluate two of its implementations on the aligned product probe $\Omega=\dyad{0}^{\otimes N}$.

\emph{Branch-recording Stinespring implementation.} This models a dynamical environment, such as a spectator spin or a scattered photon, that is flipped by the interaction itself. The unitary
\begin{equation}
    V_\theta
    =
    e^{-i\alpha Z_A\otimes Y_E}\,(U_\theta\otimes I_E),
    \qquad
    \cos\alpha = \sqrt{1-p},
    \quad
    \sin\alpha = \sqrt{p},
\end{equation}
acts on $\ket{\psi}_A\ket{0}_E$ as $\sqrt{1-p}\,U_\theta\ket{\psi}\ket{0}+\sqrt{p}\,ZU_\theta\ket{\psi}\ket{1}$ and realizes $\cE_\theta$. Its single-event intervention operator is $K_\theta=\bra{0}_EV_\theta\ket{0}_E=\sqrt{1-p}\,U_\theta$, hence collectively  $\bK_\theta^{\rm St}=(\sqrt{1-p}\,U_\theta)^{\otimes N}$. On the aligned probe this gives
\begin{equation}
    \cI_\Omega(\theta) = (1-p)^{N/2}e^{-iN\theta/2},
    \qquad
    \max_\theta F_\theta = (1-p)^N\frac{N^2}{4} .
\end{equation}
For $0<p<1$, treating $N$ as continuous, the right-hand side is maximized at $N^\ast=-2/\ln(1-p)$, and for integer $N$ the optimum is a neighboring integer.

\emph{Classically controlled implementation.} This models quasi-static noise, such as a random but fixed local field or a frozen configuration of the environment, whose value is set before the probe arrives. A classical control register $\Xi_E=(1-p)\dyad{0}+p\dyad{1}$ and the controlled unitary $V_\theta=U_\theta\otimes\dyad{0}_E+ZU_\theta\otimes\dyad{1}_E$ realize the same channel $\cE_\theta$, but now $K_\theta=(1-p)U_\theta+p\,ZU_\theta$ and $\bK_\theta^{\rm cl}=K_\theta^{\otimes N}$. Evaluating this operator on the same probe, with $Z\ket{0}=\ket{0}$, gives
\begin{equation}
    \cI_\Omega(\theta) = e^{-iN\theta/2},
    \qquad
    F_\theta = \frac{N^2}{4}
\end{equation}
for every $p$ and every $\theta$.

Both implementations have dilated generator $Z_A/2$ of norm $1/2$, so both respect $F_\theta\leq N^2/4$, but only the second attains it. The difference lies in what the auxiliary system learns about the interferometer arm. On the intervention arm the Stinespring environment is rotated from $\ket{0}_E$ to $\sqrt{1-p}\ket{0}_E+\sqrt{p}\ket{1}_E$ (for the probe $\ket{0}$, on which $Z$ acts trivially), whereas on the reference arm it stays in $\ket{0}_E$. The overlap between the two arms, and with it the single-event visibility, is therefore reduced to $\sqrt{1-p}$ even though the noise does not disturb the probe state. The classical control register, by contrast, is diagonal and is left unchanged by $V_\theta$ on both arms, so it acquires no which-arm information and the fringe is independent of $p$. The interferometric metrological power is therefore not determined by the reduced channel alone. That the interferometric visibility of a channel depends on its physical implementation is known from the theory of coherently controlled channels~\cite{oi2003interference,abbott2020communication}, and the intervention operator quantifies this dependence as a Fisher information of a single spacetime state.
\end{example}

\begin{example}
\label{ex:collective}
Take $U_\theta=e^{-i\theta Z/2}$ and, at every event, the classically controlled unitary of Example~\ref{ex:two_impl},
\begin{equation}
    V_{k,\theta}
    =
    U_\theta\otimes\dyad{0}_{E_k}
    +
    ZU_\theta\otimes\dyad{1}_{E_k}.
\end{equation}
A single correlated control state
\begin{equation}
    \Xi_\bE
    =
    (1-p)\dyad{0\cdots0}_\bE
    +
    p\dyad{1\cdots1}_\bE
\end{equation}
then gives
\begin{equation}
    \bK_\theta^{\rm coll}
    =
    (1-p)U_\theta^{\otimes N}
    +
    p(ZU_\theta)^{\otimes N}.
    \label{eqn:collective_K}
\end{equation}
Collective control of this kind arises from common-mode noise, for instance a global field that takes one random value shared by all events in each run. The resulting operator is generally nonfactorizing, although every $V_{k,\theta}$ is event local, and it differs from the intervention operator of independent local fluctuations with the same single event probabilities,
\begin{equation}
    \bK_\theta^{\rm ind}
    =
    \bigotimes_{k=1}^N
    \qty[(1-p)U_\theta+pZU_\theta],
\end{equation}
which is the operator $\bK_\theta^{\rm cl}$ of Example~\ref{ex:two_impl}. On the aligned probe $\Omega=\dyad{0}^{\otimes N}$ the two operators act identically because $Z\ket{0}=\ket{0}$, and Eq.~\eqref{eqn:collective_K} gives $\cI_\Omega(\theta)=e^{-iN\theta/2}$ and $F_\theta=N^2/4$, as for independent registers. The probe $\Omega'=\dyad{+}^{\otimes N}$, for which $\mel{+}{U_\theta}{+}=\cos(\theta/2)$ and $\mel{+}{ZU_\theta}{+}=-i\sin(\theta/2)$, separates them:
\begin{equation}
    \Tr[\bK_\theta^{\rm coll}\Omega']
    =
    (1-p)\cos^N(\theta/2)+p\,(-i)^N\sin^N(\theta/2),
    \qquad
    \Tr[\bK_\theta^{\rm ind}\Omega']
    =
    \qty[(1-p)\cos(\theta/2)-ip\sin(\theta/2)]^N .
\end{equation}
Already for $N=2$ the in-phase fringes differ by $p(1-p)\cos\theta$. Thus the aligned probe retains the spacetime Heisenberg limit fringe for the collectively controlled implementation, while a probe such as $\Omega'$ distinguishes collective from independent noise directly through the intervention operator, without constructing a selected global channel.
\end{example}

\section{Time state metrology with unresolved encounter order}
\label{sec:chronology}

\subsection{Directed histories and their classical mixture}
\label{subsec:directed}
Enumerating $N$ events as $\bA=A_1\cdots A_N$ does not fix their encounter order. For example, the realized history may be $A_2 \to A_1 \to A_3$. Let a \emph{directed history} be a definite encounter order $\pi\in S_N$ of one persistent system represented at the event labels $A_k$. The systems $E_k$ collect unobserved memory degrees of freedom at the corresponding encounters, and $\cW^{(\pi)}_k$ denotes the parameter-independent between-event channel from $A_{\pi(k)}E_{\pi(k)}$ to $A_{\pi(k+1)}E_{\pi(k+1)}$. The corresponding \emph{directed time state} $\Omega_\pi$ on $\bA$ is obtained by iterating the left product~\cite{fullwood2022quantum,lie2024unique,lie2025multi} defined in \emph{Time states and $\star$-products} of the End Matter,
\begin{equation} \label{eqn:directed_time_state}
    \Omega_\pi = \Tr_\bE \qty[\cW^{(\pi)}_{N-1}\star_L \qty(\cW^{(\pi)}_{N-2}\star_L \cdots \star_L \qty(\cW^{(\pi)}_1 \star_L \rho_{A_{\pi(1)}E_{\pi(1)}}))],
\end{equation}
and yields the interference term
\begin{equation} \label{eqn:directed}
    \cI_{\Omega_\pi}\qty(\bigotimes_{k=1}^N V_k)
    =
    \Tr\qty[
        V_{\pi(N)}\,\cW^{(\pi)}_{N-1}\qty(V_{\pi(N-1)}\cdots\cW^{(\pi)}_1\qty(V_{\pi(1)}\,\rho_{A_{\pi(1)}E_{\pi(1)}})\cdots)
    ] .
\end{equation}
The left product is used here instead of the FP product to reproduce the ordered correlator in Eq.~\eqref{eqn:directed}, including its possible time reversal, or to be more accurate, time ordering asymmetry.

Different directed histories place the same labeled events in different causal orders, but the corresponding time states act on the same labeled event space, so the convex mixture
\begin{equation} \label{eqn:mixture}
    \Omega = \sum_{\pi} p_\pi\,\Omega_\pi
\end{equation}
is well defined for a $\theta$-independent probability distribution $p_\pi$ and is again a regular spacetime state, since the defining condition $\abs{\cI_\Omega(\bW_\bA)}\leq 1$ is preserved under convex combinations of $\Omega$. Equation~\eqref{eqn:mixture} is a classical mixture of definite directed histories. It is neither a coherent quantum switch, which would require an order label kept coherent and interfered, nor a genuinely indefinite causal structure. All examples below concern the classical mixture unless stated otherwise. The spacetime state, its marginals, the interference term, and the fringe are linear in $\Omega$, so no individual encounter order has to be selected and recomposed. Fisher information is evaluated only after the mixed fringe is formed.

\subsection{Saturation conditions from the marginals of a directed time state}
\label{subsec:directed_saturation}

The saturation conditions of the main text, restated in Remark~\ref{rem:saturation_conditions}, refer to the dressed marginals $\Omega_{k,\theta}=\Tr_{\setminus k}[\bU_\theta\Omega]$, which depend on the signal. For time states built from directed histories, a sufficient condition can be stated on the marginals $\Omega_k=\Tr_{\setminus k}\Omega$ of the undressed state alone. This rests on a property of the left product. The marginal of a directed time state at an event is the reduced state of the system at that encounter, and the signals enter Eq.~\eqref{eqn:directed} as left multiplications on these states. We allow continuously differentiable signals satisfying $\partial_\theta U_{k,\theta}=-ih_{k,\theta}U_{k,\theta}$ and write $\Lambda_k(\theta):=\norm{h_{k,\theta}}_\infty$ and $\Lambda_{\rm tot}(\theta):=\sum_k\Lambda_k(\theta)$.

\begin{lemma}
\label{lem:directed_marginals}
Let $\Omega_\pi$ be the directed time state~\eqref{eqn:directed_time_state}, and let $\varrho_1=\rho_{A_{\pi(1)}E_{\pi(1)}}$ and $\varrho_{k+1}=\cW^{(\pi)}_k(\varrho_k)$, so that $\varrho_k$ is the joint state of the system and its memory on $A_{\pi(k)}E_{\pi(k)}$ at the $k$th encounter. Writing $E$ for the memory $E_{\pi(k)}$ at the encounter in question, the marginal of $\Omega_\pi$ at the event $A_{\pi(k)}$ is the reduced state of the system at that encounter,
\begin{equation}
    \Tr_{\setminus\pi(k)}\Omega_\pi=\Tr_{E}\,\varrho_k ,
\end{equation}
and is therefore a density operator. Consequently every marginal of the mixture~\eqref{eqn:mixture} is a density operator as well.
\end{lemma}

\begin{proof}
Set $V_j=I$ for every $j\neq\pi(k)$ and $V_{\pi(k)}=B$ in Eq.~\eqref{eqn:directed}. The channels $\cW^{(\pi)}_k,\ldots,\cW^{(\pi)}_{N-1}$ applied after the $k$th encounter are trace preserving, so the right-hand side reduces to $\Tr[(B\otimes I_E)\varrho_k]$. Since $B$ is arbitrary, the marginal is $\Tr_E\varrho_k$. The statement for the mixture follows by linearity.
\end{proof}

\begin{prop}
\label{prop:directed_saturation}
Let $\cJ\subseteq\bR$ be a parameter interval and fix $\theta_0\in\cJ$. Let $\Omega=\sum_\pi p_\pi\Omega_\pi$ be a classical mixture of directed time states as in Eq.~\eqref{eqn:mixture}, with arbitrary initial state, memory, and parameter-independent between-event channels. Suppose that there are phases $\phi_{k,0}\in\bR$ and a sign $\sigma\in\{+1,-1\}$ such that the undressed marginals satisfy
\begin{equation} \label{eqn:undressed_extremal}
    U_{k,\theta_0}\Omega_k=e^{-i\phi_{k,0}}\Omega_k,
    \qquad
    h_{k,\vartheta}\Omega_k
    =
    \sigma\Lambda_k(\vartheta)\Omega_k
    \quad\text{for every }\vartheta\in\cJ,
    \qquad
    \Omega_k=\Tr_{\setminus k}\Omega ,
\end{equation}
for every $k$. Thus the system arrives at each event in a fixed occupied subspace contained in the aligned extremal eigenspace of the generator throughout $\cJ$, and the signal at the reference point acts on that occupied subspace by a scalar. Define
\begin{equation} \label{eqn:accumulated_extremal_phase}
    \phi_k(\theta)
    :=
    \phi_{k,0}
    +
    \sigma\int_{\theta_0}^{\theta}\Lambda_k(\vartheta)\,d\vartheta .
\end{equation}
Then, for every $\theta\in\cJ$,
\begin{equation} \label{eqn:directed_saturation}
    \cI_\Omega(\bU_\theta)
    =
    \exp\left[-i\sum_{k=1}^N\phi_k(\theta)\right],
    \qquad
    h_{k,\theta}\Omega_{k,\theta}
    =
    \sigma\Lambda_k(\theta)\Omega_{k,\theta}
    \quad\text{for every }k,
    \qquad
    F_\theta(\Omega)=\Lambda_{\rm tot}(\theta)^2 .
\end{equation}
The last equality holds whenever $\abs{S_\theta(\Omega)}<1$, and at points where $\abs{S_\theta(\Omega)}=1$ by continuous extension whenever the prescribed limit exists. In particular the dressed marginals satisfy the extremal-marginal condition in Remark~\ref{rem:saturation_conditions} together with maximal visibility, and the pointwise spacetime Heisenberg limit of Corollary~\ref{coro:Heisenberg_limit} is saturated at every such point in $\cJ$. For a time evolution signal normalized by $U_{k,0}=I$, one may take $\theta_0=0$ and $\phi_{k,0}=0$.
\end{prop}

\begin{proof}
Let $P_k$ be the support projector of $\Omega_k$, which is a density operator by Lemma~\ref{lem:directed_marginals}. Since $\Omega_k$ is invertible on its support, Eq.~\eqref{eqn:undressed_extremal} implies
\begin{equation}
    U_{k,\theta_0}P_k=e^{-i\phi_{k,0}}P_k,
    \qquad
    h_{k,\vartheta}P_k
    =
    \sigma\Lambda_k(\vartheta)P_k
    \quad\text{for every }\vartheta\in\cJ.
\end{equation}
Moreover, $\Omega_k=\sum_\pi p_\pi\Tr_{\setminus k}\Omega_\pi$ is a convex combination of density operators, so the support of $\Tr_{\setminus k}\Omega_\pi$ is contained in the range of $P_k$ for every $\pi$ with $p_\pi>0$.

For every $k$, both $U_{k,\theta}P_k$ and $e^{-i\phi_k(\theta)}P_k$ solve the operator differential equation
\begin{equation}
    \partial_\theta X_\theta=-ih_{k,\theta}X_\theta
\end{equation}
with the same value at $\theta_0$. Uniqueness therefore gives
\begin{equation} \label{eqn:scalar_action_on_support}
    U_{k,\theta}P_k=e^{-i\phi_k(\theta)}P_k
    \quad\text{for every }\theta\in\cJ.
\end{equation}
No commutativity among the generators at different parameter values is needed for this scalar action on the fixed support.

It now suffices to prove the first two claims of Eq.~\eqref{eqn:directed_saturation} for a single directed history $\pi$, because the interference term and the dressed marginals are linear in $\Omega$. Fix such a $\pi$ and let $\varrho_k$ be as in Lemma~\ref{lem:directed_marginals}, so that $\Tr_E\varrho_k$ is supported in the range of $P_{\pi(k)}$. A joint state whose reduced state is supported in a subspace is itself supported in that subspace tensored with the memory, $\varrho_k=(P_{\pi(k)}\otimes I_E)\varrho_k(P_{\pi(k)}\otimes I_E)$. Equations~\eqref{eqn:undressed_extremal} and~\eqref{eqn:scalar_action_on_support} therefore give
\begin{equation} \label{eqn:scalar_action}
    (U_{\pi(k),\theta}\otimes I_E)\,\varrho_k
    =
    e^{-i\phi_{\pi(k)}(\theta)}\varrho_k ,
    \qquad
    (h_{\pi(k),\theta}\otimes I_E)\,\varrho_k
    =
    \sigma\Lambda_{\pi(k)}(\theta)\varrho_k .
\end{equation}
Now evaluate Eq.~\eqref{eqn:directed} with $V_j=U_{j,\theta}$ from the inside out. Define $X_1=(U_{\pi(1),\theta}\otimes I_E)\varrho_1$ and $X_{k}=(U_{\pi(k),\theta}\otimes I_E)\,\cW^{(\pi)}_{k-1}(X_{k-1})$, so that $\cI_{\Omega_\pi}(\bU_\theta)=\Tr X_N$. A scalar multiple of $\varrho_{k-1}$ is mapped by $\cW^{(\pi)}_{k-1}$ to the same multiple of $\varrho_k$, so Eq.~\eqref{eqn:scalar_action} and induction give
\begin{equation}
    X_k
    =
    \exp\left[-i\sum_{j=1}^{k}\phi_{\pi(j)}(\theta)\right]\,\varrho_k .
\end{equation}
Taking $k=N$ yields the first claim of Eq.~\eqref{eqn:directed_saturation}, which does not depend on $\pi$.

For the dressed marginal at $A_{\pi(k)}$, insert an arbitrary operator $B$ after the signal at the $k$th encounter,
\begin{equation}
    \Tr[B\,\Omega_{\pi(k),\theta}]
    =
    \Tr\qty[(U_{\pi(N),\theta}\otimes I_E)\,\cW^{(\pi)}_{N-1}\qty(\cdots\cW^{(\pi)}_{k}\qty((B\otimes I_E)X_k)\cdots)] .
\end{equation}
Replacing $B$ by $Bh_{\pi(k),\theta}$ multiplies $X_k\propto\varrho_k$ from the left by $h_{\pi(k),\theta}\otimes I_E$, which by Eq.~\eqref{eqn:scalar_action} produces the factor $\sigma\Lambda_{\pi(k)}(\theta)$. Hence $\Tr[Bh_{\pi(k),\theta}\Omega_{\pi(k),\theta}]=\sigma\Lambda_{\pi(k)}(\theta)\Tr[B\,\Omega_{\pi(k),\theta}]$ for every $B$, which is the second claim of Eq.~\eqref{eqn:directed_saturation}.

Finally, set $\Phi(\theta)=\sum_k\phi_k(\theta)$. Then $S_\theta(\Omega)=\cos\Phi(\theta)$ and $\partial_\theta S_\theta(\Omega)=-\sigma\Lambda_{\rm tot}(\theta)\sin\Phi(\theta)$, so $F_\theta(\Omega)=\Lambda_{\rm tot}(\theta)^2$ wherever $\abs{S_\theta(\Omega)}<1$ and at $\abs{S_\theta(\Omega)}=1$ by continuous extension whenever the prescribed limit exists.
\end{proof}

Proposition~\ref{prop:directed_saturation} assumes no restriction on the parameter-independent between-event channels besides trace preservation. They may be noisy and may carry memory through the $E_k$, and the encounter order may be unknown, as long as the resulting undressed marginal at each event is supported in one fixed occupied subspace contained in the aligned extremal eigenspace throughout $\cJ$. Alignment only at a single operating point, or with an occupied subspace that moves with $\theta$, is not sufficient in general because the preceding path-ordered signal evolution need not preserve that subspace. The condition can be checked on the reduced states at the encounters, which are the marginals of $\Omega$ by Lemma~\ref{lem:directed_marginals}, without dressing $\Omega$ by the signal. For a spacelike product state the analogous statement is immediate, and the content of the proposition is that the chain of channels in Eq.~\eqref{eqn:directed} does not spoil the scalar signal action on the fixed occupied subspaces.

\subsection{Common extremal eigenspace and the order-agnostic Heisenberg limit}
\label{subsec:orderHL}

Lifted to distinct tensor factors, the generators $h_k^{[k]}$ with different $k$ commute automatically. The hypothesis of the following proposition instead concerns their matrix representatives on the \emph{single-system Hilbert space} represented at each encounter, on which they need not commute. It is the special case of Proposition~\ref{prop:directed_saturation} in which one vector is extremal for every generator and the dynamics keeps the system there.

\begin{prop}
\label{prop:order_agnostic}
Let the signal at event $A_k$ be $U_{k,\theta} = e^{-i\theta h_k}$, and assume that the generators share a common extremal eigenvector on the single-system Hilbert space,
\begin{equation}
    h_k\ket{v} = \sigma\norm{h_k}_\infty\ket{v},
    \qquad
    \sigma\in\{+1,-1\} \text{ common to all } k,
\end{equation}
while the $h_k$ need not commute with one another. Let the system start in $\rho = \dyad{v}$ and let every between-event channel be a unitary channel preserving $\ket{v}$ up to a phase. Then, for an arbitrary distribution $p_\pi$ over encounter orders,
\begin{equation}
    \cI_\Omega(\bU_\theta) = e^{-i\sigma\theta\Lambda_{\rm tot}},
    \qquad
    F_\theta(\Omega) = \Lambda_{\rm tot}^2 ,
\end{equation}
with $\Lambda_{\rm tot}=\sum_k\norm{h_k}_\infty$ as in Eq.~\eqref{eqn:Lambda_tot}, saturating the spacetime Heisenberg limit of the main text at every $\theta$.
\end{prop}

\begin{proof}
In every directed history the system state is $\dyad{v}$ at each encounter, because the between-event unitaries preserve $\ket{v}$ up to a phase. By Lemma~\ref{lem:directed_marginals} every marginal of the mixture~\eqref{eqn:mixture} equals $\dyad{v}$. Taking $\cJ=\bR$, $\theta_0=0$, and $\phi_{k,0}=0$, these marginals satisfy Eq.~\eqref{eqn:undressed_extremal} with the common sign $\sigma$ for every $k$. Proposition~\ref{prop:directed_saturation} gives the claim.
\end{proof}

Hence saturation is controlled by the occupied subspace. Every signal acts by a scalar phase on the common extremal eigenspace, which we call the extremal \emph{sector}. In block form,
\begin{equation}
    h_k = \sigma\norm{h_k}_\infty\,\dyad{v} + M_k,
    \qquad
    \norm{M_k}_\infty\leq\norm{h_k}_\infty ,
\end{equation}
where $M_k$ acts on the orthogonal complement of $\ket{v}$ and the norm condition is required for $\ket{v}$ to remain extremal.

\begin{example}
\label{ex:qutrit}
Take
\begin{equation}
    h_x = \dyad{0} + \tfrac{1}{2}\qty(\dyad{1}{2} + \dyad{2}{1}),
    \qquad
    h_z = \dyad{0} + \tfrac{1}{2}\qty(\dyad{1}{1} - \dyad{2}{2}).
\end{equation}
Then $[h_x, h_z]\neq 0$ on the single-system Hilbert space, while $\norm{h_x}_\infty = \norm{h_z}_\infty = 1$ and $\ket{0}$ is a common maximal eigenvector. Both orders give
\begin{equation}
    \mel{0}{e^{-i\theta h_z}e^{-i\theta h_x}}{0}
    =
    \mel{0}{e^{-i\theta h_x}e^{-i\theta h_z}}{0}
    =
    e^{-2i\theta},
\end{equation}
hence $F_\theta = 4 = \Lambda_{\rm tot}^2$ for either order and for every mixture. The dynamics is genuinely order-sensitive on the orthogonal sector, and only the accumulated interferometric phase is an invariant of the ordering.
\end{example}

\subsection{Control pulses, metrological modes, and order readout}
\label{subsec:control_modes}

We now turn to the estimation of several parameters $\theta_1,\theta_2,\ldots,\theta_N$. Let a common generator $h$, with $\Lambda = \norm{h}_\infty>0$, act at every event through $U_{k,\theta_k} = e^{-i\theta_k h}$, and let known control pulses $W$ with
\begin{equation} \label{eqn:flip}
    W^\dag h W = -h
\end{equation}
be inserted at known positions in the encounter sequence. Condition~\eqref{eqn:flip} requires $h$ and $-h$ to be unitarily equivalent, that is, the spectrum of $h$ must be symmetric under sign reversal including multiplicities, $\operatorname{mult}_h(\mu) = \operatorname{mult}_h(-\mu)$. For a traceless qubit generator this is automatic, and a $\pi$ pulse about an orthogonal Bloch axis provides $W$. Note that if $h\ket{v} = \sigma\Lambda\ket{v}$ then $h\,W\ket{v} = -\sigma\Lambda\,W\ket{v}$, so the flipped system state is again in the extremal sector, now of $-h$.

The pulses divide the visits into intervals, which we call \emph{epochs}, and the assignment of events to epochs defines a \emph{control word} $\boldsymbol{\epsilon}\in\{+1,-1\}^N$, where $\epsilon_k$ is the sign of the phase accumulated at $A_k$. Two situations arise. If the control word is known, it selects the parameter mode that the interferometer measures, and the order of the events within each epoch is irrelevant. If two events are separated by a pulse but their assignment to the epochs is itself random, the in-phase fringe still estimates the differential mode at the Heisenberg limit, while the out-of-phase fringe records the order statistics. We treat the two situations in turn; common mode sensing, echo gradiometry, and piecewise control filters are instances of the first.

\begin{prop}
\label{prop:control_word}
Assume that the system starts in $\rho=\dyad{v}$ with $h\ket{v}=\sigma\Lambda\ket{v}$, where $\sigma\in\{+1,-1\}$. Under trivial, or more generally extremal-sector-preserving, between-event dynamics within each epoch, meaning channels that map every state supported in the extremal eigenspace of $h$ or of $-h$ to a state supported in the same eigenspace,
\begin{equation} \label{eqn:mode}
    \cI_\Omega(\boldsymbol{\theta})
    =
    \exp[-i\sigma\Lambda\sum_{k=1}^N \epsilon_k\theta_k]
\end{equation}
for every encounter order within the prescribed epochs, and therefore for every classical mixture of such orderings. The measured parameter mode is $\Phi = \sum_k \epsilon_k\theta_k$. For a one-parameter model $\theta_k = a_k g + b_k$ with known coefficients and estimated scalar $g$,
\begin{equation}
    F_g
    =
    \Lambda^2\qty(\sum_k \epsilon_k a_k)^2
    \leq
    \Lambda^2\qty(\sum_k\abs{a_k})^2 ,
\end{equation}
with equality exactly when all nonzero $\epsilon_k a_k$ have the same sign; the right-hand side is the spacetime Heisenberg limit for $g$, since the generator at $A_k$ is $a_kh$. Treating all $\theta_k$ as unknown, the Fisher matrix is $F_{jk}=\Lambda^2\epsilon_j\epsilon_k$ wherever both outcomes have nonzero probability, and by continuity at the remaining points. It has rank one, so only the mode $\Phi$ is locally identifiable.
\end{prop}

\begin{proof}
We evaluate Eq.~\eqref{eqn:mode} on the time state itself: the pulses and the intra-epoch dynamics enter only through the between-event channels of Eq.~\eqref{eqn:directed_time_state}. Let $P_\pm$ be the projectors onto the eigenspaces of $h$ with eigenvalues $\pm\Lambda$. Condition~\eqref{eqn:flip} gives $hW=-Wh$, so the pulse exchanges the two eigenspaces, $WP_\pm W^\dag=P_\mp$. Fix the pulse positions and the control word $\boldsymbol{\epsilon}$, and let $\pi\in S_N$ be any encounter order compatible with them, so that the $j$th encounter $A_{\pi(j)}$ falls in an epoch of sign $\epsilon_{\pi(j)}$. This history is the directed time state $\Omega_\pi$ of Eq.~\eqref{eqn:directed_time_state} with trivial memory, whose between-event channel $\cW^{(\pi)}_j$ is the extremal-sector-preserving dynamics of the epoch, composed with the flip channel $\Ad_W=W(\cdot)W^\dag$ whenever a pulse is applied between the $j$th and $(j+1)$th encounters. By Eq.~\eqref{eqn:directed},
\begin{equation} \label{eqn:control_chain}
    \cI_{\Omega_\pi}(\boldsymbol{\theta})
    =
    \Tr\qty[
        U_{\pi(N),\theta_{\pi(N)}}\,\cW^{(\pi)}_{N-1}\qty(U_{\pi(N-1),\theta_{\pi(N-1)}}\cdots\cW^{(\pi)}_1\qty(U_{\pi(1),\theta_{\pi(1)}}\,\dyad{v})\cdots)
    ] .
\end{equation}
Let $\varrho_1=\dyad{v}$ and $\varrho_{j+1}=\cW^{(\pi)}_j(\varrho_j)$ be the system state at the $j$th encounter in the absence of the signals. The initial state is supported in the range of $P_\sigma$, an extremal-sector-preserving channel keeps the support in the current eigenspace, and each pulse moves it to the opposite one, so $\varrho_j$ is supported in the range of $P_{\sigma\epsilon_{\pi(j)}}$, on which $U_{\pi(j),\theta_{\pi(j)}}=e^{-i\theta_{\pi(j)}h}$ acts as a scalar,
\begin{equation} \label{eqn:control_scalar_action}
    U_{\pi(j),\theta_{\pi(j)}}\,\varrho_j
    =
    e^{-i\sigma\Lambda\epsilon_{\pi(j)}\theta_{\pi(j)}}\,\varrho_j .
\end{equation}
Evaluating Eq.~\eqref{eqn:control_chain} from the inside out as in the proof of Proposition~\ref{prop:directed_saturation}, each signal contributes its scalar by Eq.~\eqref{eqn:control_scalar_action} and each channel carries a scalar multiple of $\varrho_j$ to the same multiple of $\varrho_{j+1}$, so
\begin{equation}
    \cI_{\Omega_\pi}(\boldsymbol{\theta})
    =
    \prod_{j=1}^N e^{-i\sigma\Lambda\epsilon_{\pi(j)}\theta_{\pi(j)}}\,\Tr\varrho_N
    =
    \exp[-i\sigma\Lambda\Phi] ,
\end{equation}
where $\Tr\varrho_N=1$ because the channels are trace preserving. The global phase of $W$ never appears, since the pulse enters the time state only through the channel $\Ad_W$. The exponent is $\sum_j\epsilon_{\pi(j)}\theta_{\pi(j)}=\sum_k\epsilon_k\theta_k=\Phi$ for every encounter order compatible with the epochs, because such orders only permute the events within epochs. Hence Eq.~\eqref{eqn:mode} holds for every such directed time state and, by linearity of the interference term in $\Omega$, for every classical mixture~\eqref{eqn:mixture} of them. The fringe is $S_{\boldsymbol{\theta}}(\Omega)=\cos(\Lambda\Phi)$. For the one-parameter model, $\partial_gS=-\Lambda\qty(\sum_k\epsilon_ka_k)\sin(\Lambda\Phi)$, so $F_g=(\partial_gS)^2/(1-S^2)=\Lambda^2\qty(\sum_k\epsilon_ka_k)^2$ wherever $\abs{S}<1$, and by continuity elsewhere; the bound follows from $\abs{\sum_k\epsilon_ka_k}\leq\sum_k\abs{a_k}$, with equality exactly when all nonzero $\epsilon_ka_k$ have the same sign. Treating all $\theta_k$ as unknown, $\partial_{\theta_j}S=-\Lambda\epsilon_j\sin(\Lambda\Phi)$, which gives $F_{jk}=\Lambda^2\epsilon_j\epsilon_k$.
\end{proof}

Several familiar settings follow at once. The choice $\epsilon_k = +1$ for all $k$ is common mode sensing of $\Theta := \sum_k\theta_k = N\bar\theta$, with $F_{\bar\theta} = N^2\Lambda^2$. A single sign change gives echo gradiometry of a difference, and several fixed-sign epochs form a piecewise control filter. Let $N_\pm$ denote the numbers of events assigned to the two epoch types. If a balanced constraint $N_+ = N_-$ is imposed, for instance to cancel a common offset $b_k = \bar\theta$, the equality condition above may be unattainable. For even $N$ the constrained optimum then assigns one epoch to the largest $N/2$ coefficients $a_k$, equivalently, up to an overall sign, to the smallest $N/2$.
A coarse control pattern thus selects a metrological mode, while arbitrary unresolved permutations within each control epoch, and any classical mixture over them, leave the corresponding Fisher information unchanged.

We now let the assignment itself be random. For a mixture of two directed histories, let $\cI_{ba}$ denote the interference term of the history in which $A_a$ is encountered before $A_b$, including any pulses and between-event dynamics, and let $\cI_{ab}$ be that of the opposite history. If $p$ is the probability of the former, linearity of the mixed spacetime state gives
\begin{equation}
\begin{aligned}
    \cI_p
    &=
    p\cI_{ba}+(1-p)\cI_{ab}
    \\
    &=
    \cI_+ +(2p-1)\cI_-,
    \label{eqn:order_decomposition}
\end{aligned}
\end{equation}
where $\cI_\pm=(\cI_{ba}\pm\cI_{ab})/2$. The order distribution enters only through the antisymmetric part $\cI_-$, which vanishes when the two histories are indistinguishable, for instance for commuting signals with trivial between-event dynamics, and which is erased by the exchange-symmetric mixture $p=1/2$. In the echo setting the signals commute, and the antisymmetric part is produced entirely by the pulse, whose position in the sequence is fixed while the events that precede and follow it are not.

\begin{example}
\label{ex:echo}
Take the signals $U_{k,\theta_k}=e^{-i\theta_k h}$ with $\Lambda=\norm{h}_\infty$, an aligned initial state $\rho=\dyad{v}$ satisfying $h\ket{v}=\Lambda\ket{v}$, and a flip pulse $W$ satisfying Eq.~\eqref{eqn:flip}. For two events $A_1$ and $A_2$ separated by this pulse, the directed histories give $\cI_{21}(\Delta)=e^{-i\Lambda\Delta}$ and $\cI_{12}(\Delta)=e^{i\Lambda\Delta}$, where $\Delta=\theta_1-\theta_2$. Equation~\eqref{eqn:order_decomposition} becomes
\begin{equation}
    \cI_p(\Delta)
    =
    \cos(\Lambda\Delta)
    -
    i(2p-1)\sin(\Lambda\Delta).
    \label{eqn:echo_fringe}
\end{equation}
Writing $\theta_{1,2}=\bar\theta\pm\Delta/2$, the generator norm with respect to $\Delta$ is $\Lambda/2$ at each event, so the spacetime Heisenberg limit is $F_\Delta\leq\Lambda^2$. The in-phase fringe gives
\begin{equation}
    S_\Delta
    =
    \cos(\Lambda\Delta),
    \qquad
    F_\Delta
    =
    \Lambda^2,
\end{equation}
saturating the bound for every $p$. The complementary fringe
\begin{equation}
    S_\Delta^{(\pi/2)}
    =
    -(2p-1)\sin(\Lambda\Delta)
\end{equation}
records the order imbalance. If $\Delta$ is calibrated, its Fisher information for $p$ is
\begin{equation}
    F_p^{(\pi/2)}
    =
    \frac{4\sin^2(\Lambda\Delta)}
    {1-(2p-1)^2\sin^2(\Lambda\Delta)},
\end{equation}
which reaches $1/[p(1-p)]$ when $\abs{\sin(\Lambda\Delta)}=1$. The same mixed spacetime state therefore supports order-robust metrology of $\Delta$ and, in the complementary quadrature, estimation of the classical order distribution. If both $p$ and $\Delta$ are unknown, they form a joint estimation problem. The common mode $\bar\theta$ drops out.
\end{example}

\end{document}